\documentclass[a4paper,11pt]{article}

\usepackage[a4paper]{geometry}
\usepackage{color}
\usepackage{graphicx}
\usepackage{amsmath}
\usepackage{amsfonts}
\usepackage{amsthm}
\usepackage{amscd}
\usepackage{epsfig}
\usepackage{amssymb}
\usepackage{slashed}
\usepackage{tabularx}
\usepackage{calligra}
\usepackage{enumerate}
\usepackage{hyperref}
\usepackage{float}
\usepackage{stackrel}
\usepackage[T1]{fontenc}
\usepackage{longtable}
\usepackage{amsthm}
\usepackage{mathrsfs}
\usepackage{verbatim} 
\usepackage{fancyhdr}
\usepackage{tikz}
\usepackage{tikz-cd}
\usepackage{slashed}
\usetikzlibrary{matrix}
\usetikzlibrary{positioning}
\usepackage[all]{xy}
\usepackage{units}
\usepackage{textcomp}
\usepackage{turnstile}
\usepackage{vmargin}
\usepackage{anysize}
\usepackage{ stmaryrd }
\usepackage{tikz-3dplot}

\setmargins{2,5cm}       
{1,5cm}                        
{16cm}                      
{23,42cm}                    
{10pt}                           
{1cm}                           
{0pt}                             
{2cm}                           

\theoremstyle{plain}
\newtheorem{theorem}{Theorem}[section]
\newtheorem{proposition}[theorem]{Proposition}

\newtheorem{remark}[theorem]{Remark}

\numberwithin{equation}{section}

\def\d{{\rm d}}
\def\i{{\rm i}}
\def\CP{\mathbb{CP}}

\begin{document}

\title{\bf Hidden symmetries of generalised gravitational instantons}

\author{Bernardo Araneda\footnote{Email: \texttt{bernardo.araneda@aei.mpg.de}}  \\
Max-Planck-Institut f\"ur Gravitationsphysik \\ 
(Albert-Einstein-Institut), Am M\"uhlenberg 1, \\
D-14476 Potsdam, Germany}

\date{\today}

\maketitle

\begin{abstract}
For conformally K\"ahler Riemannian four-manifolds with a Killing field, we present a framework to solve the field equations for generalised gravitational instantons corresponding to conformal self-duality and to cosmological Einstein-Maxwell. After deriving generic identities for the curvature of such manifolds without assuming field equations, we obtain $SU(\infty)$ Toda formulations for the Page-Pope, Pleba\'nski-Demia\'nski, and Chen-Teo classes, we show how to solve the (modified) Toda equation, and we use this to find conformally self-dual and Einstein-Maxwell generalisations of these geometries.
\end{abstract}

\section{Introduction}
\label{Sec:Intro}

Gravitational instantons \cite{Hawking1976, GH1979, Gibbons1980} are four-dimensional, complete, Ricci-flat Riemannian manifolds with sufficiently fast curvature decay. They are expected to give the dominant contributions to the path integral for Euclidean quantum gravity.
Particular cases are metrics with self-dual Riemann tensor, while more general cases correspond to generalisations of the Ricci-flat condition. These generalisations include the addition of a cosmological constant, solutions to Einstein-Maxwell theory, conformally self-dual geometries (i.e. metrics with self-dual Weyl tensor), Bach-flat metrics, etc. These geometries, which we will call {\em generalised gravitational instantons} here, are the main object of study in this paper.

The curvature decay condition means that the geometry is required to have at most quartic volume growth. This leads to a number of possibilities such as asymptotically flat (AF), asymptotically locally flat (ALF), asymptotically Euclidean (AE) or asymptotically locally Euclidean (ALE) metrics. In quantum gravity, these have a physical interpretation in terms of vacuum and finite-temperature states of the gravitational field \cite{Gibbons1980, GPR1979}. In the Ricci-flat case, the only AE instanton is flat space, but the ALE class includes non-trivial examples such as the Eguchi-Hanson instanton. An ALF example is the Taub-NUT instanton, while the AF class contains for instance the Euclidean Schwarzschild and Kerr solutions, as well as the recently discovered Chen-Teo instanton \cite{ChenTeo1}, which gives a counterexample to the Euclidean black hole uniqueness conjecture \cite{Gibbons1980}. We mention that other types of asymptotic structures compatible with the required volume growth are also possible; these are termed ALG, ALH, ALG*, ALH*. 

Under the hyper-K\"ahler assumption, a complete classification of gravitational instantons is now available \cite{CC1, CVZ, Cherkis1, Cherkis2, Hein1, Hein2, Kronheimer, Kronheimer2, Minerbe, Sun}. Recall that the hyper-K\"ahler condition means that the metric is K\"ahler with respect to a 2-sphere of complex structures, and that, assuming that the manifold is simply connected, this is equivalent to (anti-)self-duality of the Riemann tensor. (In the non-simply connected case this equivalence is no longer true, see \cite{Wright}.)
If one assumes only the K\"ahler condition, then Ricci-flatness implies ({\em locally}) hyper-K\"ahler. A less-restrictive assumption is the Hermitian condition. In fact, as noticed in \cite[Question 1.4]{AkAnd}, all known examples of Ricci-flat gravitational instantons are Hermitian. Using Bianchi identities, this implies that they are {\em conformally} K\"ahler, and have at least one Killing field. If, in addition, one also assumes the existence of a second (commuting) Killing field, then the recent Biquard-Gauduchon classification \cite{BG, BG2} proves that the only toric, Ricci-flat, Hermitian, ALF instantons are the Kerr, Chen-Teo, Taub-bolt and Taub-NUT metrics. While the Kerr and Taub-bolt/NUT metrics have non-Ricci-flat (cosmological and Einstein-Maxwell) versions, the corresponding generalisation of the Chen-Teo instanton is an open problem, and constitutes a major motivation for the current paper.

From \cite{DunajskiTod}, a four-dimensional geometry is conformally K\"ahler if and only if it 
has a solution to the 2-index twistor equation, or, in tensor terms, a conformal Killing-Yano 2-form. These objects have appeared in gravitational physics since at least the work of Walker and Penrose \cite{WalkerPenrose} on the geometric origin of Carter's constant in the Kerr space-time. Many more geometries admitting these tensors have since been found, and, in a general relativity context, these Killing-like objects are now usually referred to as {\em hidden symmetries}. They also give origin to other ``ordinary'' symmetries (Killing vectors). We refer to \cite{Frolov} for a recent review on this topic.

Motivated by the above, in this work we study `generalised gravitational instantons with hidden symmetries', that is, Riemannian, non-Ricci-flat 4-manifolds with a conformal K\"ahler structure and a Killing field. We will focus on the field equations corresponding to conformal self-duality and cosmological Einstein-Maxwell. 
Our methods are based on the remarkable constructions developed by LeBrun \cite{Lebrun91} (scalar-flat K\"ahler metrics with symmetry), Tod \cite{Tod2020} (one-sided type D Ricci-flat metrics), and Flaherty \cite{Flaherty78} (connections between the scalar-flat K\"ahler and Einstein-Maxwell systems).
Our main results are as follows:

\begin{itemize}
\item We obtain, in section \ref{Sec:CKinstantons}, a number of identities for conformally K\"ahler metrics whose Ricci tensor is invariant under the complex structure (but no other assumptions), which give a compact formula for the Ricci form, and allow to construct generalised instantons in a variety of gravitational theories from solutions to a single scalar ($SU(\infty)$ Toda) equation.
\item In section \ref{Sec:PPclass} we show that the Page-Pope ansatz \cite{PP} is generically ambi-K\"ahler without assuming any field equations, and we classify all conformally self-dual solutions. It follows that some recently considered \cite{Giribet} Bach-flat instantons are actually conformally Einstein. We also obtain the general cosmological Einstein-Maxwell solution.
\item In section \ref{Sec:PD} we obtain the $SU(\infty)$ Toda formulation for the Pleba\'nski-Demia\'nski ansatz, and we give a simple trick to solve the (modified) Toda equation. We use this to construct a conformally self-dual Pleba\'nski-Demia\'nski space, which depends on 5 parameters and is not Einstein, so it is different from the standard self-dual limit of Pleba\'nski-Demia\'nski.
\item In section \ref{Sec:ChenTeo} we obtain the $SU(\infty)$ Toda formulation for the Chen-Teo ansatz, and we give a simple trick to solve the corresponding Toda equation. We construct a family of (non-Ricci-flat) conformally self-dual Chen-Teo geometries, as well as a family of Einstein-Maxwell Chen-Teo solutions.
\end{itemize}

We end this introduction with some further motivation and references. Non-Ricci-flat instantons are interesting in a number of situations in high-energy physics, for example: Einstein-Maxwell theory coincides with the bosonic sector of $N=2$ supergravity in four dimensions; conformally self-dual and Bach-flat geometries arise in conformal gravity; and Hawking's approach to space-time foam does not require the vacuum field equations \cite{Hawking1978, HPP1979}. Examples of cosmological Einstein-Maxwell instantons have been studied in \cite{Dunajski06, Dunajskietal1, Gutowski, Dunajskietal2}, while instantons in conformal gravity were considered in \cite{Strominger, Smilga} and more recently in \cite{DunajskiTod2014, Li2012, Lu2012, Liu2012, Giribet}. Recall that the conformal gravity field equations are the vanishing of the Bach tensor. 
Since in this work we are interested in conformally K\"ahler geometry, and Derdzi\'nski showed in  \cite[Proposition 4]{Derdzinski} that a K\"ahler metric with non-self-dual Weyl tensor is Bach-flat {\em if and only if} it is (locally) conformally Einstein, we will then not worry about solutions to the Bach-flat equations {\it per se} (as any such metric would be either conformally self-dual or a conformal transformation of an Einstein metric). In fact, we note that a classification of compact, Bach-flat K\"ahler surfaces has been recently given in \cite{Lebrun2020}. Other notable results in the non-Ricci-flat case that are related to the current paper include for instance compact Einstein-Hermitian \cite{Lebrun2012}, self-dual Einstein \cite{ApostolovGauduchon, Calderbank}, AE/ALE K\"ahler \cite{HeinLebrun}, and ambi-toric geometries \cite{ApostolovCalderbankGauduchon}. The $SU(\infty)$ Toda system was also discovered by Tod to play a key role for (anti-)self-dual Einstein metrics \cite{Tod95b}. Finally, Tod's recent work \cite{Tod2024} shows that one-sided type-D Einstein-Maxwell metrics with two commuting Killing fields can be encoded in a pair of axisymmetric solutions of the flat three-dimensional Laplacian.

\section{Conformally K\"ahler geometry}
\label{Sec:CKinstantons}

\subsection{Preliminaries}
\label{Sec:HS}

Let $(M,g_{ab})$ be a conformally K\"ahler 4-manifold\footnote{Our constructions will be purely local. In particular, some of our examples include the 4-sphere $S^4$, which does not admit a global complex structure.}, with complex structure $J^{a}{}_{b}$ and fundamental 2-form $\kappa_{ab}=g_{bc}J^{c}{}_{a}$. Recall that $\kappa_{ab}$ is necessarily self-dual (SD) or anti-self-dual (ASD) w.r.t. to the Hodge star. 
We will follow the conventions of Penrose \& Rindler \cite{PR1, PR2}: according to these, ASD 2-forms are expressed in spinor notation as $\varphi_{AB}\epsilon_{A'B'}$, while SD 2-forms are $\psi_{A'B'}\epsilon_{AB}$; see Eqs. (3.4.38)-(3.4.39) and the footnote in page 151 in \cite{PR1}. 
We choose to work with unprimed spinor indices, which implies that we choose $\kappa_{ab}$ to be ASD w.r.t the above convention. (Unfortunately, this means that our orientation conventions are opposite to the standard ones \cite{Besse, Derdzinski} in complex and Riemannian geometry.)
Then $\kappa_{ab}=j_{AB}\epsilon_{A'B'}$, where $j_{AB}$ is symmetric and satisfies $j^{A}{}_{C}j^{C}{}_{B}=-\delta^{A}_{B}$. 
The conformally rescaled 2-form is $\hat\kappa_{ab}=\Omega j_{AB}\hat\epsilon_{A'B'}$, where $\hat\epsilon_{A'B'}=\Omega\epsilon_{A'B'}$. The conformal K\"ahler property is $\hat{\nabla}_{a}\hat\kappa_{bc}=0$, where $\hat{\nabla}_{a}$ is the Levi-Civita connection of $\hat{g}_{ab}$. In spinors, this translates into $\hat\nabla_{AA'}(\Omega j_{BC})=0$. Using the relation between $\hat\nabla_{AA'}$ and $\nabla_{AA'}$ (see \cite{PR1, PR2, DunajskiTod}), one deduces that $(M,g_{ab})$ has a solution to the 2-index twistor equation:
\begin{align}
\nabla_{A'(A}K_{BC)}=0, \qquad K_{AB}:=\Omega^{-1} j_{AB}. \label{KillingSpinor}
\end{align}
Define now $Z_{ab}:=K_{AB}\epsilon_{A'B'}$. A calculation using the 2-index twistor equation (see \cite[Eq. (6.4.6)]{PR2}) shows that
\begin{align}\label{CKYE}
 \nabla_{a}Z_{bc} =\nabla_{[a}Z_{bc]} - 2g_{a[b}\xi_{c]}, \qquad \xi_{a}:=\tfrac{1}{3}\nabla^{b}Z_{ab}.
\end{align}
The first equation is the conformal Killing-Yano (CKY) equation. In terms of the fundamental 2-form, the CKY tensor is $Z_{ab}=\Omega^{-1}\kappa_{ab}$.

Notice that $\xi_{b}$ has always zero divergence, $\nabla^{a}\xi_{a}=0$ (this follows from $\nabla^{a}\nabla^{b}Z_{ab}=0$ since $Z_{ab}$ is a 2-form). In addition, a calculation shows that $\xi_{b}$ can be expressed as
\begin{align}
 \xi_{b} = J^{a}{}_{b}\partial_{a}\Omega^{-1}. \label{vectorCKY}
\end{align}
From this expression, one deduces that the vector field $\xi^{a}\partial_{a}$ preserves both the conformal factor $\Omega$ and the fundamental 2-form $\kappa_{ab}$. 

The integrability conditions for the 2-index twistor equation equation lead to some restrictions for the curvature of $g_{ab}$. For example, the Ricci tensor $R_{ab}$ must satisfy the following identity:
\begin{align}
 R_{ab} - R_{cd}J^{c}{}_{a}J^{d}{}_{b} = 4J^{c}{}_{a}\nabla_{(c}\xi_{b)}. \label{InvCondRicci}
\end{align}
This equation implies that $\xi_{a}$ (given by \eqref{vectorCKY}) is a Killing vector if and only if the Ricci tensor is invariant under the complex structure, meaning that $R_{ab}=R_{cd}J^{c}{}_{a}J^{d}{}_{b}$.
Furthermore, the ASD Weyl tensor must be Petrov type D: decomposing $K_{AB}$ in terms of its principal spinors as $K_{AB}=2\i\,\Omega^{-1} o_{(A}o^{\dagger}_{B)}$ (where $o^{\dagger}_A$ is the spinor conjugate of $o_A$, and $o_{A}o^{\dagger A}=1$), one deduces that 
$\Psi_{ABCD}=6\Psi_{2}o_{(A}o_{B}o^{\dagger}_{C}o^{\dagger}_{D)}$, where
\begin{align}
\Psi_{2}:= \Psi_{ABCD}o^{A}o^{B}o^{\dagger C}o^{\dagger D} = -\tfrac{1}{8}C_{abcd}J^{ac}J^{bd} \label{ASDWeyl}.
\end{align}
This Weyl scalar is closely related to the scalar curvature $\hat{R}$ of the K\"ahler metric $\hat{g}_{ab}=\Omega^{2}g_{ab}$: a calculation (using the identity $\hat{\nabla}_{a}J^{b}{}_{c}=0$) shows that
\begin{align}
\Psi_{2}=\Omega^{2}\frac{\hat{R}}{12}. \label{Psi2}
\end{align}
Since this is the only non-trivial component of the ASD Weyl tensor, conformal self-duality reduces simply to the scalar equation $\Psi_2=0$.

There is also a close connection between K\"ahler geometry and the Einstein-Maxwell equations. This dates back to the work of Flaherty \cite{Flaherty78}, who showed that scalar-flat K\"ahler metrics are automatically solutions to the Einstein-Maxwell system. This remarkable connection has been generalised to the case of constant-scalar-curvature K\"ahler surfaces \cite{Lebrun2010}, and even to conformally K\"ahler metrics with $J$-invariant Ricci tensor \cite{ApostolovCalderbankGauduchon, Lebrun2015, Lebrun2016, ApostolovMaschler, Koca}. Recall that, given a 4-manifold $(M,g_{ab})$ and a 2-form $F_{ab}$, the Einstein-Maxwell equations with cosmological constant $\lambda$ are
\begin{equation}\label{CEM}
\begin{aligned}
 R_{ab}-\frac{R}{2}g_{ab} +\lambda g_{ab} ={}& 2F_{ac}F_{b}{}^{c} - \frac{1}{2}g_{ab}F_{cd}F^{cd}, \\
 \nabla^{a}F_{ab} ={}& 0 = \nabla_{[a}F_{bc]}.
\end{aligned}
\end{equation}
Suppose that $(M,g_{ab})$ is conformally K\"ahler with $J$-invariant Ricci tensor, and define $F_{ab}=F^{-}_{ab}+F^{+}_{ab}$, where 
\begin{align}
F^{-}_{ab} := \Omega^{2}\kappa_{ab}, \qquad 
F^{+}_{ab} := \tfrac{1}{4}\Omega^{-2}(\rho_{ab}-\tfrac{R}{4}\kappa_{ab}), 
\label{SDMaxwell}
\end{align}
and $\rho_{ab}=R_{bc}J^{c}{}_{a}$ is the Ricci form. Then one can show, see \cite[Proposition 5]{ApostolovCalderbankGauduchon} and \cite[Theorem A]{Lebrun2015}, that the system \eqref{CEM} is equivalent to the constancy of the scalar curvature: $R=4\lambda$. 

So far we focused on geometric objects, but one can also obtain powerful coordinate expressions for various quantities: LeBrun's description \cite{Lebrun91} of K\"ahler surfaces with a Killing field can be adapted to (strictly) conformally K\"ahler metrics (see Tod's formulation \cite{Tod2020} for the Ricci-flat case). That is, if $(M,g_{ab})$ is strictly conformally K\"ahler with a $J$-invariant Ricci tensor, then there are local coordinates $(\psi,x,y,z)$, real functions $W(x,y,z), u(x,y,z)$, and a 1-form $A(x,y,z)$ (with $\partial_{\psi}\lrcorner A=0$) such that the metric $g$ and the fundamental 2-form $\kappa$ can be written in {\em LeBrun normal form}:
\begin{align}
 g ={}& W^{-1}(\d\psi+A)^{2} + W[\d{z}^{2}+e^{u}(\d{x}^{2}+\d{y}^{2})], \label{metric-gen} \\
 \kappa ={}& (\d\psi+A)\wedge\d{z} + We^{u}\d{x}\wedge\d{y}, \label{F2form-gen}
\end{align}
cf. \cite{Lebrun91, ACG, Tod2020}.
In the rest of this paper we will make heavy use of these expressions, so it is convenient to give a formula for the quantities in \eqref{metric-gen}: using the objects in \eqref{vectorCKY}, and letting $\ell,m$ be two independent type-(1,0) forms, we have
\begin{align}
 \xi=\partial_{\psi}, \qquad
 z=\Omega^{-1}, \qquad 
 W^{-1}=g_{ab}\xi^{a}\xi^{b}, \qquad 
 e^{u/2}(\d{x}+\i\d{y}) = 2 \, \xi \, \lrcorner \, (\ell\wedge m). 
 \label{KeyVariables}
\end{align}
The 1-form $A$ is defined by lowering an index of $\xi^a$, i.e. $\xi_{a}\d{x}^{a}=W^{-1}(\d\psi+A)$. The expressions for the K\"ahler metric $\hat{g}=\Omega^2 g$ and K\"ahler form $\hat\kappa = \Omega^2 \kappa$ are identical to \eqref{metric-gen} and \eqref{F2form-gen} after replacing $(z,W,e^{u})$ by $(\hat{z},\hat{W},e^{\hat{u}})$, where
\begin{align}
\hat{z} := -\frac{1}{z}, \qquad 
\hat{W} := z^{2}W, \qquad 
e^{\hat{u}} := \frac{e^u}{z^4}.
\label{hattedvariables}
\end{align}

The integrability conditions for the complex structure, together with $\d(z^{-2}\kappa)=0$, lead to the following `generalised' monopole equation relating the unknowns $A,W,u$:
\begin{align}
 \d{A} = -z^2\partial_{z}(\tfrac{We^u}{z^2})\d{x}\wedge\d{y}+(\partial_{y}W)\d{x}\wedge\d{z} - (\partial_{x}W) \d{y}\wedge\d{z}. \label{monopoleEq}
\end{align}
Applying an additional exterior derivative, one gets
\begin{align}
 W_{xx}+W_{yy}+\partial_{z}\left[z^2\partial_{z}(\tfrac{We^u}{z^2})\right]=0, \label{EqForW}
\end{align}
which involves only two unknowns, $W$ and $u$. To find an expression that involves only one of the unknowns, one can look at the scalar curvature: a calculation shows that the Ricci scalars of the metrics $\hat{g}$ and $g$ are, respectively,
\begin{align}
 \hat{R} ={}& -\frac{1}{\hat{W}e^{\hat{u}}}\left[ \hat{u}_{xx} + \hat{u}_{yy} + (e^{\hat{u}})_{\hat{z}\hat{z}} \right], \label{Todaghat} \\
  R ={}& -\frac{1}{We^u}\left[ u_{xx} + u_{yy} + (e^{u})_{zz} \right]. \label{Todag} 
\end{align}
We note that, when proving these identities \eqref{Todaghat}-\eqref{Todag}, while one uses that $\partial_{\psi}$ is a Killing vector, one does not use that the Ricci tensor of $g$ is $J$-invariant (i.e. even though $\partial_{\psi}$ is Killing, the vector field \eqref{vectorCKY} need not be Killing in order for \eqref{Todaghat}-\eqref{Todag} to hold; this observation is relevant for section \ref{sec:PPCSD}, see below eq. \eqref{candidatesKilling}).

\subsection{The Ricci form}

Having given the expressions \eqref{metric-gen}, \eqref{F2form-gen}, \eqref{KeyVariables}, the main result of this section is the following:

\begin{theorem}\label{Result:Ricci}
Let $(M,g_{ab},\kappa_{ab})$ be a conformally K\"ahler Riemannian 4-manifold whose Ricci tensor is invariant under the complex structure, so that $\xi_a$ given by \eqref{vectorCKY} is a Killing field and $g_{ab}$ and $\kappa_{ab}$ have the expressions \eqref{metric-gen}-\eqref{F2form-gen}. Then the Ricci form $\rho_{ab}=R_{bc}J^{c}{}_{a}$ is 
\begin{align}
\rho = \frac{1}{2}We^uR \: \d{x}\wedge\d{y} - \frac{W}{z^2}\left[ \tilde{*}\d 
- \xi \wedge \d \right]\left(\frac{W_{0}}{W}\right), 
\label{RicciFormFormula2}
\end{align}
where $R$ is the Ricci scalar \eqref{Todag}, we defined 
\begin{align}
 W_{0}:=z\left(1-\frac{zu_z}{2} \right) \label{Wring}
\end{align}
and, for an arbitrary function $f$, the operator $\tilde{*}\d$ is
\begin{align}
 \tilde{*}\d f := f_x \d{y}\wedge\d{z} + f_{y} \d{z}\wedge\d{x} + e^{u}f_z\d{x}\wedge\d{y}.
 \label{tilde*d}
\end{align}
In addition, the trace-free Ricci form can be expressed as
\begin{equation}
\begin{aligned}
\rho - \tfrac{R}{4}\kappa ={}& \left[\tfrac{R}{4}-\tfrac{1}{z^2}\partial_{z}(\tfrac{W_{0}}{W}) \right](-(\d\psi+A)\wedge\d{z} + We^{u}\d{x}\wedge\d{y}) \\
& + \frac{W}{z^2} \left[ \xi\wedge(\d{x} \, \partial_{x} + \d{y} \, \partial_{y}) 
- \d{z}\wedge (\d{x} \, \partial_{y}-\d{y} \, \partial_{x}) \right](\tfrac{W_{0}}{W}).
\label{TraceFreeRF}
\end{aligned}
\end{equation}
\end{theorem}

\begin{remark}
\begin{enumerate}
\item From \eqref{RicciFormFormula2} we see that Ricci-flatness $\rho_{ab}=0$ reduces to $\frac{W_{0}}{W}=\gamma={\rm const.}$ together with the $SU(\infty)$ Toda equation $u_{xx} + u_{yy} + (e^{u})_{zz}=0$, so we recover Tod's result \cite{Tod2020}. Comparison to the Schwarzschild case (cf. section \ref{Sec:Examples} below) suggests to use the notation $\gamma\equiv -M$, where $M$ is Schwarzschild's mass.
\item From \eqref{TraceFreeRF}, the Einstein condition $\rho-\tfrac{R}{4}\kappa=0$ (i.e. $R_{ab}=\lambda g_{ab}$) is satisfied if and only if $u$ satisfies the modified Toda equation
\begin{align}
 u_{xx}+u_{yy}+(e^{u})_{zz}=-4\lambda W e^{u}
\end{align}
and $\frac{W_{0}}{W}$ is a function of only $z$ satisfying $\frac{1}{z^2}\frac{\d}{\d{z}}(\tfrac{W_{0}}{W})=\lambda$, whose solution is
\begin{align}
 W\equiv W_{\lambda}:=\frac{W_{0}}{\frac{\lambda}{3}z^3+\gamma}=\frac{z\left(1-\frac{zu_z}{2} \right)}{\frac{\lambda}{3}z^3+\gamma}
 \label{Wlambda}
\end{align}
where $\gamma$ is an integration constant. This was also obtained by Tod in \cite{Tod2020}.
\item In the cosmological Einstein-Maxwell case, $R=4\lambda$, in view of \eqref{SDMaxwell}, formula \eqref{TraceFreeRF} gives us an explicit expression for the SD part of the Maxwell field.
\end{enumerate}
\end{remark}

\begin{proof}[Proof of Theorem \ref{Result:Ricci}]
We start by recalling that for any two metrics $g_{ab}$ and $\hat{g}_{ab}=\Omega^{2}g_{ab}$, whose Ricci tensors are $R_{ab}$ and $\hat{R}_{ab}$ respectively, the relation between $R_{ab}$ and $\hat{R}_{ab}$ is given by a standard conformal transformation formula (we use \cite[Eq. (D.8)]{Wald} in four dimensions):
\begin{align*}
\hat{R}_{ab}=R_{ab}-2\nabla_a\nabla_b\log\Omega+2(\nabla_a\log\Omega)(\nabla_b\log\Omega)-g_{ab}g^{cd}[\nabla_c\nabla_d\log\Omega+2(\nabla_c\log\Omega)(\nabla_d\log\Omega)].
\end{align*}
Since we are interested in the case in which there is a complex structure which is parallel w.r.t. the Levi-Civita connection $\hat\nabla_{a}$ of $\hat{g}_{ab}$, and we also have the special relation \eqref{vectorCKY} involving derivatives of $\Omega^{-1}$, it is convenient to rewrite the above formula as
\begin{gather}
 R_{ab} = \hat{R}_{ab} - 2\Omega\hat\nabla_a\hat\nabla_b\Omega^{-1} + 4\Omega^{2}(\hat\nabla_a\Omega^{-1})(\hat\nabla_b\Omega^{-1}) - \Sigma \hat{g}_{ab}, \label{CTRicci} \\
 \Sigma := \Omega\hat{g}^{ab}\hat\nabla_a\hat\nabla_b\Omega^{-1} + \Omega^2\hat{g}^{ab}(\hat\nabla_a\Omega^{-1})(\hat\nabla_b\Omega^{-1}). \label{Sigma}
\end{gather}
(To quickly check \eqref{CTRicci}-\eqref{Sigma}, simply eliminate the logarithm in the formula for $\hat{R}_{ab}$ above and then formally replace $R_{ab}\leftrightarrow\hat{R}_{ab}$, $\nabla \to \hat\nabla$, $\Omega\to\Omega^{-1}$.)
Assuming now that $\hat{g}_{ab}$ is K\"ahler, we recall from the previous section that $\hat\nabla_{a}J^{b}{}_{c}=0$, $\hat\nabla_a\Omega^{-1}=-\xi_{b}J^{b}{}_{a}$, where $\xi_a$ is defined in \eqref{CKYE}. At this point we are not assuming that $\xi_{a}$ is Killing. Contracting  \eqref{CTRicci} with $J^{b}{}_{c}$, we get
\begin{align}
R_{ab}J^{b}{}_{c} = -\hat{\rho}_{ac}-2\Omega\hat\nabla_a\xi_c-4\Omega^2J^{b}{}_{a}\xi_{b}\xi_{c}+\Sigma\hat\kappa_{ac},  \label{RicciJ0}
\end{align}
where $\hat{\rho}_{ac}\equiv \hat{R}_{cb}J^{b}{}_{b}$ is the Ricci form of $\hat{g}_{ab}$. From \eqref{InvCondRicci}, we know that the Ricci tensor $R_{ab}$ is $J$-invariant if and only if $\xi_{a}$ is Killing, $\nabla_{(a}\xi_{b)}=0$. Assuming this to be the case, $R_{ab}J^{b}{}_{c}$ is anti-symmetric and so we can define the Ricci form of $g_{ab}$, $\rho_{ac}:=R_{cb}J^{b}{}_{a}=-\rho_{ca}$. Using that the two terms with $\xi_{a}$ in \eqref{RicciJ0} are anti-symmetric in $ac$ (not separately but together), after some manipulations we get the formula 
\begin{align}
\rho = \hat{\rho} - \Sigma\hat\kappa +\Omega^{-1}\d(\Omega^2\xi). \label{RFformula0}
\end{align}
Now, we want to express \eqref{RFformula0} in terms of $u,W$. For the scalar $\Sigma$, defined in \eqref{Sigma}, note that it can be written as $\Sigma = \Omega\hat\Box\Omega^{-1}+W^{-1}$, where $\hat\Box=\hat{g}^{ab}\hat\nabla_{a}\hat\nabla_{b}$. Alternatively, we have $\Omega\hat\Box\Omega^{-1}=-\Omega^{-3}\Box\Omega=-z^{3}\Box z^{-1}$. We then find 
\begin{align}
 \Sigma = \frac{zu_z - 1}{W}. \label{Sigma1}
\end{align}
For the Ricci form $\hat{\rho}$ of the K\"ahler metric $\hat{g}_{ab}$, we use the well-known formula $\hat{\rho}=-\i\partial\bar\partial\log\det(\hat{g}_{\alpha\bar\beta})$, where $\partial=\d{z}^{\alpha}\wedge\partial_{\alpha}$ and  $\bar\partial=\d\bar{z}^{\alpha}\wedge\partial_{\bar{\alpha}}$ are Dolbeault operators defined by the complex structure. We have $\log\det(\hat{g}_{\alpha\bar\beta}) = \hat{u}$, so $\hat{\rho}=-i\partial\bar\partial\hat{u}$. In addition, the identity $-i\partial\bar\partial f = \frac{1}{2}\d(J\d f)$ holds for any smooth $f$. Putting then $f=\hat{u}$ and using $J(\d{x})=-\d{y}$, $J(\d{y})=\d{x}$, $J(\d{z})=\xi$, we get
\begin{align*}
\hat{\rho} = \tfrac{1}{2}\left[-(\hat{u}_{xx}+\hat{u}_{yy})\d{x}\wedge\d{y}-\hat{u}_{xz}\d{z}\wedge\d{y}+\hat{u}_{yz}\d{z}\wedge\d{x}+\d\hat{u}_{z}\wedge\xi+\hat{u}_{z}\d\xi\right].
\end{align*}
Replacing this expression, together with \eqref{Sigma1} and \eqref{F2form-gen}, in equation \eqref{RFformula0}:
\begin{align*}
\rho ={}& \left[ -\tfrac{1}{2}(\hat{u}_{xx}+\hat{u}_{yy})-\tfrac{(zu_z-1)}{z^2}e^{u} \right]\d{x}\wedge\d{y}-\tfrac{1}{2}\hat{u}_{xz}\d{z}\wedge\d{y}+\tfrac{1}{2}\hat{u}_{yz}\d{z}\wedge\d{x} \\
& +\left[ \tfrac{1}{2}\d\hat{u}_{z}+\tfrac{(zu_z-1)}{z^2}\d{z}-\tfrac{2}{z^2}\d{z} \right]\wedge\xi + \left( \tfrac{1}{2}\hat{u}_{z}+\tfrac{1}{z} \right)\d\xi.
\end{align*}
Using the relation \eqref{hattedvariables} between $\hat{u}$ and $u$, after some tedious computations we arrive at the unenlightening expression
\begin{align*}
\rho ={}& -\tfrac{1}{2}\left[u_{xx}+u_{yy}+\tfrac{2(zu_z-1)}{z^2}e^u+\tfrac{2z}{W}(\tfrac{zu_z}{2}-1)\partial_{z}(\tfrac{We^u}{z^2}) \right]\d{x}\wedge\d{y} \\
& +\tfrac{1}{2}\left[ u_{yz}-\tfrac{2}{zW}(\tfrac{zu_z}{2}-1)W_{y} \right](\d{z}\wedge\d{x}+\d{y}\wedge\d\xi)\\
& +\tfrac{1}{2}\left[ u_{xz}-\tfrac{2}{zW}(\tfrac{zu_z}{2}-1)W_{x} \right](-\d{z}\wedge\d{y}+\d{x}\wedge\d\xi)\\
& +\tfrac{1}{2}\left[\tfrac{1}{z^2}(z^2u_{zz}+2zu_z-2)-\tfrac{2}{zW}(\tfrac{zu_z}{2}-1)W_{z} \right]\d{z}\wedge\xi.
\end{align*}
Now, defining $W_{0}$ as in \eqref{Wring}, we have the identities
\begin{align*}
 u_{yz}-\tfrac{2}{zW}\left(\tfrac{zu_z}{2}-1\right)W_{y} ={}& -\tfrac{2W}{z^2}\partial_{y}\left(\tfrac{W_{0}}{W}\right), \\
 u_{xz}-\tfrac{2}{zW}\left(\tfrac{zu_z}{2}-1\right)W_{x} ={}& -\tfrac{2W}{z^2}\partial_{x}\left(\tfrac{W_{0}}{W}\right), \\
\tfrac{1}{z^2}(z^2u_{zz}+2zu_z-2)-\tfrac{2}{zW}(\tfrac{zu_z}{2}-1)W_{z} ={}& -\tfrac{2W}{z^2}\partial_{z}\left(\tfrac{W_{0}}{W}\right),
\end{align*}
which lead to 
\begin{align*}
\rho ={}& -\tfrac{1}{2}\left[u_{xx}+u_{yy}+ \tfrac{2e^u}{z^2}(1-\tfrac{W_{0}}{W}W_z - W_{0}u_z) \right]\d{x}\wedge\d{y} \\
& - \tfrac{W}{z^2}\d{z}\wedge(\d{x} \:\partial_{y} - \d{y} \:\partial_{x})(\tfrac{W_{0}}{W}) - \tfrac{W}{z^2}\d(\tfrac{W_{0}}{W})\wedge\xi. 
\end{align*}
Defining the operator $\tilde{*}\d$ as in \eqref{tilde*d}, we have 
\begin{align*}
\d{z}\wedge(\d{x} \:\partial_{y}-\d{y} \:\partial_{x})(\tfrac{W_{0}}{W}) = \tilde{*}\d(\tfrac{W_{0}}{W}) - e^{u}\partial_{z}(\tfrac{W_{0}}{W})\d{x}\wedge\d{y},
\end{align*}
which then leads to our final formula \eqref{RicciFormFormula2}. Having shown this, the proof of \eqref{TraceFreeRF} requires only a few more tedious but straightforward computations, so we will omit them. 
\end{proof}

In addition to the Ricci scalar \eqref{Todag} and the Ricci form \eqref{RicciFormFormula2}, we will also need expressions for the ASD Weyl tensor. Recalling that the only non-trivial component is $\Psi_2$, and using \eqref{Psi2}, \eqref{Todaghat} and \eqref{Wring}, a short calculation gives
\begin{align}
 \Psi_{2} = -\frac{1}{z^3}\frac{W_{0}}{W} + \frac{R}{12}. \label{Psi2-2}
\end{align}
In particular, notice that in the Ricci-flat and Einstein cases, we recover a well-known relation between $\Psi_{2}$ and the conformal factor: $\Psi_{2}\propto\Omega^{3}$ (recall $\Omega=z^{-1}$).

Finally, we also notice that a geometry may be conformally K\"ahler w.r.t. {\em both} ASD and SD orientations: this is called an ambi-K\"ahler structure \cite{ApostolovCalderbankGauduchon}. In this case we have two integrable complex structures $(J_{\pm})^{a}{}_{b}$ and two K\"ahler metrics $\hat{g}^{\pm}_{ab}=\Omega^{2}_{\pm}g_{ab}$. As in \eqref{vectorCKY}, we now have 
\begin{align}
\xi^{\pm}_{b}=(J_{\pm})^{a}{}_{b}\partial_{a}\Omega^{-1}_{\pm}. \label{CandKAK}
\end{align}
If at least one of $\xi^{a}_{\pm}$ is a Killing vector, then all of the previous results apply w.r.t. the corresponding orientation $\pm$. If both $\xi^{a}_{\pm}$ are Killing, then we will have two Toda formulations, one for each orientation: the corresponding Toda variables $(u_{\pm},W_{\pm},z_{\pm},x_{\pm},y_{\pm})$ are the analogue of \eqref{KeyVariables},
\begin{align}
 z_{\pm} = \Omega_{\pm}^{-1}, \qquad 
 W_{\pm}^{-1} = g_{ab}\xi^{a}_{\pm}\xi^{b}_{\pm}, \qquad 
 e^{u_{\pm}/2}(\d{x}_{\pm}+\i\d{y}_{\pm}) = 2 \, \xi_{\pm} \, \lrcorner \, (\ell\wedge m^{\pm}),
 \label{TVAK}
\end{align}
where in the last equality we defined $m^{+}\equiv m$, $m^{-}\equiv \bar{m}$. Analogously to \eqref{Wring}, we also put $W^{\pm}_{0}=z_{\pm}(1-\frac{1}{2}z_{\pm}\partial_{z_{\pm}}u_{\pm})$. The Weyl tensor is now type $D\otimes D$. The only non-trivial components are $\Psi_{2}^{-}\equiv\Psi_{2}$ and $\Psi_{2}^{+}\equiv\tilde{\Psi}_{2}$, which can be computed as
\begin{align}
 \Psi_{2}^{\pm} = \Omega_{\pm}^{2}\frac{\hat{R}_{\pm}}{12}
 =-\frac{1}{z^{3}_{\pm}}\frac{W^{\pm}_0}{W_{\pm}} + \frac{R}{12}, 
 \label{Psi2AK}
\end{align}
where $\hat{R}_{\pm}$ are the Ricci scalars of the K\"ahler metrics $\hat{g}^{\pm}_{ab}$. (Recall that \eqref{Psi2AK} is valid regardless of whether \eqref{CandKAK} are Killing or not.)

\subsection{Examples}
\label{Sec:Examples}

\noindent
{\bf Spherical symmetry.}
Consider the ansatz
\begin{align}
 g = f(r)\d\tau^2+\frac{\d{r}^2}{f(r)}+r^{2}(\d\theta^2+\sin^2\theta\d\varphi^2), \label{SS}
\end{align}
where $f$ is an arbitrary smooth function of $r$. Choose the coframe $\beta^0=f^{1/2}\d\tau$, $\beta^1=f^{-1/2}\d{r}$, $\beta^2=r\d\theta$, $\beta^3=r\sin\theta\d\varphi$, define the fundamental 2-forms $\kappa^{\pm}=\beta^{0}\wedge\beta^{1}\mp\beta^{2}\wedge\beta^{3}$ ($\kappa^{+}$ is SD and $\kappa^{-}$ is ASD), and the associated almost-complex structures $(J_{\pm})^{a}{}_{b}=\kappa^{\pm}_{bc}g^{ca}$. The type-$(1,0)$ eigenspaces of $J_{\pm}$ are generated by $\ell=\frac{1}{\sqrt{2}}(\beta^{0}+\i\beta^{1})$, $m^{\pm}=\frac{1}{\sqrt{2}}(\beta^2\mp\i\beta^3)$. Putting $a^{0}_{\pm}=f^{-1/2}$, $b^{0}_{\pm}=0$, and $a^{1}_{\pm}=0$, $b^{1}_{\pm}=(r\sin\theta)^{-1}$, we see that the type $(1,0)$-forms $a^{\alpha}_{\pm}\ell+b^{\alpha}_{\pm}m^{\pm}$ ($\alpha=0,1$) are closed, so both $J_{+}$ and $J_{-}$ are integrable. Furthermore, if $\Omega_{\pm}\equiv r^{-1}$ then it is straightforward to see that $\d[\Omega_{\pm}^{2}\kappa^{\pm}]=0$, so the geometry \eqref{SS} is ambi-K\"ahler. Finally, computing the vector fields \eqref{CandKAK}, we get $\xi^{a}_{\pm}\partial_{a}=\partial_{\tau}$, which is Killing. Thus, regardless of the form of the arbitrary function $f(r)$, the geometry is conformally K\"ahler (in fact ambi-K\"ahler) with a $J$-invariant Ricci tensor. We then compute the variables \eqref{TVAK}:
\begin{align}
z_{\pm}=r, \qquad W_{\pm}^{-1}=f, \qquad e^{u_{\pm}} = fr^2\sin^2\theta, \qquad \d{x_{\pm}}=\frac{\d\theta}{\sin\theta}, \qquad \d{y_{\pm}}=\mp\d\varphi. \label{TodaSS}
\end{align}
Using the formulas of previous sections, a short calculation gives
\begin{align}
 R = \frac{2 - (r^2 f)''}{r^{2}}, \qquad
 \Psi^{\pm}_{2}= \frac{\{ 2-r^2[r^2(f/r^2)' ]' \}}{12 r^2}, \qquad
 \frac{W^{\pm}_{0}}{W_{\pm}} = -\frac{r^2 f' }{2},
\end{align}
where a prime $'$ represents a derivative w.r.t. $r$. 

The cosmological Einstein-Maxwell equations are $R=4\lambda$, which gives
\begin{align}
 f(r)=1+\frac{a_1}{r}+\frac{a_2}{r^2}-\frac{\lambda}{3}r^2
\end{align}
for some constants $a_1,a_2$. The Weyl scalars $\Psi^{\pm}_2$ and the two pieces $F^{\pm}$ of the Maxwell field are:
\begin{align}
 \Psi^{\pm}_{2} =-\frac{a_{1}}{2r^3}-\frac{a_{2}}{r^{4}}, \qquad 
 F^{\pm} = (-a_{2}/4)^{\frac{1\pm1}{2}}
 \left(\frac{1}{r^2}\d\tau\wedge\d{r} \mp \sin\theta\d\theta\wedge\d\varphi \right).
\end{align}
Setting $a_{1}\equiv-2M$, $a_{2}\equiv Q^2$, we recognise the Euclidean Reissner-N\"ordstrom-(A)dS solution.

We can alternatively look for $f(r)$ such that the ansatz \eqref{SS} is conformally self-dual. Recall that this is equivalent to $\Psi^{-}_{2}\equiv\Psi_{2}=0$. Since $\Psi^{-}_{2}=\Psi^{+}_{2}$, we see that $\Psi_{2}=0$ gives $C_{abcd}\equiv0$, so the self-dual solution to the ansatz \eqref{SS} is conformally flat. The condition $\Psi^{\pm}_{2}=0$ gives $f(r)=1+b_1 r+ b_2 r^2$ for arbitrary constants $b_1,b_2$. The Ricci scalar and trace-free Ricci form are 
\begin{align}
R=-6\left(\frac{b_{1}}{r}+2b_2\right), \qquad 
\rho^{\pm}-\frac{R}{4}\kappa^{\pm}=\frac{2b_1}{r}\left(\d\tau\wedge\d{r}\pm r^2\sin\theta\d\theta\wedge\d\varphi\right).
\end{align}
We see that the geometry is Einstein iff $b_1=0$, in which case it reduces to Euclidean (anti-)de Sitter space with cosmological constant $-3b_{2}$ (which is $S^4$ if $b_2<0$).

\medskip
\noindent
{\bf The Kerr-Newman ansatz.}
This example allows us to illustrate a trick to solve the Toda equation, that will be key for later sections. Consider the ansatz 
\begin{align}
 g = \frac{\Delta}{\Sigma}(\d\tau-a\sin^2\theta\d\varphi)^2+\frac{\sin^2\theta}{\Sigma}[a\d\tau+(r^2-a^2)\d\varphi]^2+\frac{\Sigma}{\Delta}\d{r}^2+\Sigma\d\theta^2, \label{KNAnsatz}
\end{align}
where $a$ is a real constant, $\Sigma=r^2-a^2\cos^2\theta$, and $\Delta=\Delta(r)$ is an arbitrary function of $r$. As in the previous example, we start by choosing a coframe: $\beta^0=(\frac{\Delta}{\Sigma})^{1/2}(\d\tau-a\sin^2\theta\d\varphi)$, $\beta^1=(\frac{\Sigma}{\Delta})^{1/2}\d{r}$, $\beta^2=\sqrt{\Sigma}\d\theta$, $\beta^3=\frac{\sin\theta}{\sqrt{\Sigma}}[a\d\tau+(r^2-a^2)\d\varphi]$; we define $\kappa^{\pm}=\beta^{0}\wedge\beta^{1}\mp\beta^{2}\wedge\beta^{3}$ and the almost-complex structures $(J_{\pm})^{a}{}_{b}=\kappa^{\pm}_{bc}g^{ca}$. Type-$(1,0)$ forms for $J_{\pm}$ are spanned by $\ell=\frac{1}{\sqrt{2}}(\beta^{0}+\i\beta^{1})$ and $m^{\pm}=\frac{1}{\sqrt{2}}(\beta^2\mp\i\beta^3)$. Putting $a^{0}_{\pm}=\frac{(r^2-a^2)}{\sqrt{\Delta\Sigma}}$, $b^{0}_{\pm}=\frac{\pm\i a\sin\theta}{\sqrt{\Sigma}}$, and $a^{1}_{\pm}=\frac{\pm\i a}{\sqrt{\Delta\Sigma}}$, $b^{1}_{\pm}=\frac{1}{\sin\theta\sqrt{\Sigma}}$, we see that $a^{\alpha}_{\pm}\ell+b^{\alpha}_{\pm}m^{\pm}$ ($\alpha=0,1$) are closed, so $J_{+}$ and $J_{-}$ are integrable. Defining $\Omega_{\pm}=(r\mp a\cos\theta)^{-1}$, a short calculation shows that $\d[\Omega^{2}_{\pm}\kappa^{\pm}]=0$, so the metric \eqref{KNAnsatz} is ambi-K\"ahler. The vector fields \eqref{CandKAK} can be computed to be $\xi^{a}_{\pm}\partial_{a}=\partial_{\tau}$, so they are Killing. Therefore, independently of the form of $\Delta(r)$, the metric \eqref{KNAnsatz} is conformally K\"ahler (ambi-K\"ahler) with a $J$-invariant Ricci tensor. We can then compute the variables \eqref{TVAK}:
\begin{align}
 z_{\pm}=r\mp a\cos\theta, \quad W_{\pm}=\frac{\Sigma}{\Delta+a^2\sin^2\theta}, \quad e^{u_{\pm}}=\Delta\sin^2\theta, \quad \d{x}_{\pm}=\frac{\mp a\d{r}}{\Delta}+\frac{\d\theta}{\sin\theta}, \quad \d{y_{\pm}}=\mp\d\varphi. 
 \label{TodaKN}
\end{align}

For concreteness let us focus on $z_{-},x_{-},y_{-}$, etc., and let us omit the subscript $-$.
The Einstein-Maxwell equations with $\lambda=0$ reduce to the $SU(\infty)$ Toda equation with symmetry: $u_{xx}+(e^{u})_{zz}=0$. From \eqref{TodaKN} we find the vector fields $\partial_{x},\partial_{z}$:
\begin{align}
\partial_{x}=\tfrac{\sin\theta\Delta}{\Delta+a^2\sin^2\theta}\left(a\sin\theta\partial_{r}+\partial_{\theta}\right), \qquad \partial_{z}=\tfrac{1}{\Delta+a^2\sin^2\theta}\left(\Delta\partial_{r}-a\sin\theta\partial_{\theta}\right). \label{VFKN}
\end{align}
One can check that $u_{xx}+(e^{u})_{zz}=0$ then becomes a quite complicated differential equation: we do not seem to gain anything in passing from $(x,z)$ to $(r,\theta)$. However, we can do the following trick: introduce an auxiliary variable $\sigma$ by
\begin{align}
 u_{x} = \sigma_{z}, \qquad (e^{u})_{z} = -\sigma_{x}. \label{trickKN}
\end{align}
Using \eqref{VFKN}, these equations lead respectively to:
\begin{subequations}
\begin{align}
 a\sin^2\theta\dot\Delta+2\cos\theta\Delta ={}& \Delta \: \partial_{r}\sigma - a\sin\theta \: \partial_{\theta}\sigma, \label{xz} \\
 -\sin\theta\dot\Delta+2a\sin\theta\cos\theta ={}& a\sin\theta \: \partial_{r}\sigma+\partial_{\theta}\sigma,
 \label{zx}
\end{align}
\end{subequations}
where $\dot\Delta=\frac{\d\Delta}{\d{r}}$. Now, from \eqref{zx} we get an expression for $\partial_{r}\sigma$, and we replace this in \eqref{xz}. The resulting equation relates $\dot\Delta$ and $\partial_{\theta}\sigma$. We then replace this new expression for $\partial_{\theta}\sigma$ back in \eqref{zx}. After these manipulations, we get the following system:
\begin{align}
\partial_{\theta}\sigma=-\sin\theta\dot{\Delta}, \qquad \partial_{r}\sigma=2\cos\theta.
\end{align}
Using now the identity $\partial_{r}\partial_{\theta}\sigma=\partial_{\theta}\partial_{r}\sigma$, we immediately get:
\begin{align}
 \ddot{\Delta}=2 \qquad \Rightarrow \qquad \Delta=r^2+c_{1}r+c_{2},
\end{align}
for some constants $c_{1},c_{2}$. The metric is automatically Einstein-Maxwell, but to interpret $c_1,c_2$, we compute the rest of the curvature (recall from \eqref{SDMaxwell} that the Ricci form is $\rho=4\Omega^2 F^{+}$):
\begin{align}
 \Psi^{\pm}_{2} ={}& -\frac{c_1/2}{(r\mp a\cos\theta)^{3}} 
 - \frac{(a^2+c_{2})}{(r\mp a\cos\theta)^{3}(r\pm a\cos\theta)}, \label{Psi2KN} \\
 F^{\pm} ={}& \frac{[-\frac{1}{4}(a^2+c_2)]^{\frac{1\pm 1}{2}}}{(r\mp a\cos\theta)^{2}}\left[ 
 \d\tau\wedge\d(r\mp a\cos\theta)-\sin\theta\d\varphi\wedge(a\sin\theta\d{r} \mp (r^2-a^2)\d\theta) \right]. \label{RicciKN}
\end{align}
Setting $c_{1}\equiv-2M$, $a^2+c_{2}\equiv Q^2$, we recognise the Euclidean Kerr-Newman solution.

\section{The Page-Pope class}
\label{Sec:PPclass}

\subsection{Preliminaries}

Consider a Riemann surface $\Sigma$ with a Riemannian metric $g_{\Sigma}=2h\:\d\zeta\d\bar\zeta$, where $\zeta=\frac{1}{\sqrt{2}}(x+\i y)$ is a holomorphic coordinate and $h$ is a real positive function. Let $\kappa_{\Sigma}=\i h \: \d\zeta\wedge\d\bar\zeta$ be the K\"ahler form. Since $\d\kappa_{\Sigma}=0$, there is, locally, a 1-form $A$ such that $\kappa_{\Sigma}=\d A$ and a K\"ahler potential $K_{\Sigma}$ with $h=\partial_{\zeta}\partial_{\bar\zeta}K_{\Sigma}$. We now define a manifold $M$ as the total space of a fibre bundle over $\Sigma$ with 2-dimensional fibers parametrised by $r,\psi$, and we introduce a Riemannian metric $g$ on $M$ by
\begin{align}
 g = F(r)\d{r}^2+G(r)(\d\psi+A)^2+H(r)g_{\Sigma}
 \label{PPclass}
\end{align}
where $F,G,H$ are arbitrary (non-negative) functions of $r$. Note that, by redefining the coordinate $r$, the three functions $F,G,H$ can be reduced to two. For the moment we will focus on the form \eqref{PPclass}, but we will later make use of this freedom.

The class of metrics \eqref{PPclass} includes geometries such as Fubini-Study, Eguchi-Hanson, Taub-NUT, K\"ahler surfaces of Calabi type (cf. \cite{ACG, Lucietti}), particular cases of the Bianchi IX class, etc. It is the restriction to four dimensions of the geometries considered by Page and Pope in \cite{PP}. In \cite{PP}, the conditions on the functions $F,G,H$ so that the metric \eqref{PPclass} is Einstein are determined, and they find that, under the Einstein assumption, the metric is conformal to two different K\"ahler metrics. We will first show that this ambi-K\"ahler structure is actually independent of the form of $F,G,H$, and so it is independent of field equations; then we will use this result to study generalised instantons.

\begin{proposition}
For any functions $F,G,H$, the class of metrics \eqref{PPclass} is (locally) ambi-K\"ahler.
\end{proposition}

\begin{proof}
Write the metric on the Riemann surface as $g_{\Sigma}=h(\d{x}^2+\d{y}^2)$, and choose the coframe $\beta^0=\sqrt{G}(\d\psi+A)$, $\beta^1=\sqrt{F}\d{r}$, $\beta^2=\sqrt{Hh} \: \d{x}$, $\beta^3=\sqrt{Hh} \: \d{y}$. We define the fundamental 2-forms (with opposite orientation)
\begin{align}
\kappa^{\pm}=\beta^{0}\wedge\beta^{1}\mp\beta^{2}\wedge\beta^{3}=\sqrt{FG}(\d\psi+A)\wedge\d{r} \mp H\kappa_{\Sigma},
\label{ffPP}
\end{align}
and the associated almost-complex structures $(J_{\pm})^{a}{}_{b}=\kappa^{\pm}_{bc}g^{ca}$. Type-$(1,0)$ forms for $J_{\pm}$ are $\ell=\frac{1}{\sqrt{2}}(\beta^{0}+\i\beta^{1})$, $m^{\pm}=\frac{1}{\sqrt{2}}(\beta^{2}\mp\i\beta^{3})$. Let $a^{0}_{\pm}=\frac{1}{\sqrt{G}}$, $b^{0}_{\pm}=\frac{\mp\i}{\sqrt{2Hh}}\partial_{\zeta}K_{\Sigma}$, $a^{1}_{\pm}=0$, $b^{1}_{\pm}=\frac{1}{\sqrt{Hh}}$. Then a short calculation shows that the type-$(1,0)$ forms $a^{\alpha}_{\pm}\ell+b^{\alpha}_{\pm}m^{\pm}$ ($\alpha=0,1$) are closed, so $J_{+}$ and $J_{-}$ are both integrable. Furthermore, using \eqref{ffPP} we find that regardless of the form of $F,G,H$ we always have
\begin{align}
\d[\Omega_{\pm}^{2}\kappa^{\pm}] = 0, \qquad 
\Omega_{\pm}^{2} \equiv \frac{c_{\pm}}{H}\exp\left [\pm \: \int\frac{\sqrt{FG}}{H}\d{r}\right], 
\label{CKPP}
\end{align}
where $c_{\pm}$ are arbitrary constants. Thus \eqref{PPclass} is always ambi-K\"ahler.
\end{proof}

The vector fields \eqref{CandKAK} are
\begin{align}
 \xi^{a}_{\pm}\partial_{a} = \frac{1}{\sqrt{FG}}\frac{\d(\Omega^{-1}_{\pm})}{\d{r}}\partial_{\psi}.
 \label{candidatesKilling}
\end{align}
Since $\partial_{\psi}$ is a Killing vector of \eqref{PPclass}, we see that \eqref{candidatesKilling} are in general not Killing. Requiring \eqref{candidatesKilling} to be Killing imposes restrictions on $\Omega_{\pm}$, which means restrictions on $F,G,H$.

\subsection{Conformally self-dual solutions}
\label{sec:PPCSD}

Define
\begin{align}
\hat{F}_{\pm}:=\Omega^{2}_{\pm}F, \qquad \hat{G}_{\pm}:=\Omega^{2}_{\pm}G, \qquad \hat{H}_{\pm}:=\Omega^{2}_{\pm}H.
\label{hattedquantitiesPP}
\end{align}
Using \eqref{CKPP}, we see that $\d{\hat{H}_{\pm}}=\pm\sqrt{\hat{F}_{\pm}\hat{G}_{\pm}}\d{r}$. Thus, if we further define
\begin{align}
 \hat{z}_{\pm}:=-\hat{H}_{\pm}, \qquad \hat{W}_{\pm}:=\hat{G}_{\pm}^{-1}, \qquad e^{\hat{u}_{\pm}} := \hat{G}_{\pm}\hat{H}_{\pm}h, \label{TVPP}
\end{align}
then the K\"ahler metrics $\hat{g}^{\pm}=\Omega_{\pm}^{2}g$ and K\"ahler forms $\hat{\kappa}^{\pm}=\Omega_{\pm}^{2}\kappa$ become  
\begin{align}
 \hat{g}^{\pm} ={}& \hat{W}_{\pm}^{-1}(\d\psi+A)^{2} + \hat{W}_{\pm}[\d\hat{z}_{\pm}^{2}+e^{\hat{u}_{\pm}}(\d{x}^{2}+\d{y}^{2})], \\
 \hat{\kappa}^{\pm} ={}&\mp[ (\d\psi+A)\wedge\d\hat{z}_{\pm} + \hat{W}_{\pm}e^{\hat{u}_{\pm}}\d{x}\wedge\d{y} ].
\end{align}
A straightforward calculation using \eqref{Todaghat} gives
\begin{align}
 \hat{R}_{\pm} = \frac{1}{\hat{H}_{\pm}}\left[ R_{\Sigma} +\frac{\d^2}{\d\hat{z}_{\pm}^{2}}(\hat{G}_{\pm}\hat{z}_{\pm}) \right], \qquad
 R_{\Sigma}:=-h^{-1}(\partial^{2}_{x}+\partial^{2}_{y})\log h.
 \label{RPP}
\end{align}
The conformal (A)SD equations then reduce to 
\begin{align}
\frac{\d^2}{\d\hat{z}_{\pm}^{2}}(\hat{G}_{\pm}\hat{z}_{\pm})=-R_{\Sigma}. \label{ASDeqPP}
\end{align}
The left side is a function of $\hat{z}_{\pm}$ only, while the right side is a function of $(x,y)$ only. Thus, \eqref{ASDeqPP} demands $R_{\Sigma}$ to be constant. Assuming $\Sigma$ to be simply connected, this implies that $g_{\Sigma}$ is isometric to the standard metric of either the 2-sphere ($R_{\Sigma}>0$), the Euclidean 2-plane ($R_{\Sigma}=0$), or the hyperbolic plane ($R_{\Sigma}<0$). The solution to \eqref{ASDeqPP} is $\hat{G}_{\pm}\hat{z}_{\pm} =-\tfrac{R_{\Sigma}}{2} \hat{z}^2_{\pm} + b_{\pm} \hat{z}_{\pm} - a_{\pm}$, so  
\begin{align}
 \hat{G}_{\pm} = -\frac{R_{\Sigma}}{2}\hat{z}_{\pm} + b_{\pm} - \frac{a_{\pm}}{\hat{z}_{\pm}},
\end{align}
where $a_{\pm},b_{\pm}$ are arbitrary constants. Recalling $\hat{z}_{\pm}=-\hat{H}_{\pm}$, expressing the above equation in terms of $G,H,\Omega_{\pm}$, and summarising:

\begin{theorem}\label{Result:SDPP}
The metric \eqref{PPclass} is a solution to the conformally (A)SD equations $*C_{abcd}=\mp C_{abcd}$ iff $g_{\Sigma}$ has constant curvature $R_{\Sigma}$ and the functions $F,G,H$ satisfy
\begin{align}
 G = \frac{R_{\Sigma}}{2}H + \frac{b_{\pm}}{\Omega_{\pm}^{2}} + \frac{a_{\pm}}{\Omega_{\pm}^{4}H}. 
 \label{PPSDsol}
\end{align}
where $a_{\pm},b_{\pm}$ are arbitrary constants and $\Omega_{\pm}$ are defined in \eqref{CKPP}.
\end{theorem}

\begin{remark}[Classification]\label{Remark:ClassificationPP}
It is worth mentioning a different perspective on the above solutions.
From \eqref{TVPP}, the function $\hat{u}_{\pm}$ is ``separable'' in the sense that $\hat{u}_{\pm}=f(x,y)+g(z)$, where $f(x,y)=\log h$ and $g(z)=\log(\hat{G}_{\pm}\hat{H}_{\pm})$. Thus, we solved the Toda equation (here $\hat{R}_{\pm}=0$) when the Toda variable is separable. All separable solutions to the Toda equation were classified by Tod in \cite{Tod95}: the classification is in terms of three constants $k,a,b$, which in our notation are $k\equiv-\frac{R_{\Sigma}}{2}$, $a\equiv b_{\pm}$, $b\equiv-a_{\pm}$.
\end{remark}

Let us see some examples. In all three examples that follow, we take $\Sigma=\CP^1\cong S^2$ with the round 2-metric, which has $R_{\Sigma}=2$.

\paragraph{Fubini-Study.}
Taking the functions $F = \frac{1}{(1+r^2)^2}$, $G = \frac{r^2}{4(1+r^2)^2}$, $H = \frac{r^2}{4(1+r^2)}$, the ambi-K\"ahler class \eqref{PPclass} becomes the Fubini-Study metric on $M=\mathbb{CP}^2$. Putting $c_{+}=\frac{1}{4}$, $c_{-}=4$ in \eqref{CKPP}, we find $\Omega^{2}_{+}=1$ and $\Omega^{2}_{-}=H^{-2}$, so the metric is actually K\"ahler w.r.t. $J_{+}$ (and of course conformally K\"ahler w.r.t. $J_{-}$). We have $\hat{G}_{+}=G$, $\hat{H}_{+}=H$, $\hat{G}_{-}=\frac{4}{r^2}$, $\hat{H}_{-}=\frac{4(1+r^2)}{r^2}$. This gives $\hat{z}_{\pm}=H^{\pm 1}$. Replacing in \eqref{RPP}, we find $\hat{R}_{+}=24$, $\hat{R}_{-}=0$. This is of course consistent with the fact that Fubini-Study is Einstein with cosmological constant equal to $6$ and has a self-dual Weyl tensor. Using \eqref{candidatesKilling}, we also note that $\xi^{a}_{+}$ vanishes and $\xi^{a}_{-}\partial_{a}\equiv \partial_{\psi}$ is Killing.

\paragraph{Generalised Eguchi-Hanson.}
Let $F=\frac{1}{f(r)}$, $G=\frac{r^2f(r)}{4}$, $H=\frac{r^2}{4}$, where $f(r)$ is an arbitrary function of $r$. The metric \eqref{PPclass} is then a ``generalised Eguchin-Hanson'' space. Setting $c_{+}=\frac{1}{4}$, $c_{-}=4$ in \eqref{CKPP}, we find $\Omega^{2}_{+}=1$ and $\Omega^{2}_{-}=H^{-2}$ (so the metric is K\"ahler w.r.t. $J_{+}$). The form of $f(r)$ that makes the space conformally (A)SD can now be easily found by solving the algebraic equation \eqref{PPSDsol}:
\begin{align}
 *C_{abcd}=\mp C_{abcd} \quad \iff \quad f(r) = 1+b_{\pm}\left(\frac{2}{r}\right)^{\pm 2}+a_{\pm}\left(\frac{2}{r}\right)^{\pm 4}. \label{SDGEH}
\end{align}
We also note that $\xi^{a}_{+}$ vanishes and $\xi^{a}_{-}\partial_{a}\equiv \partial_{\psi}$ is Killing, so the rest of the curvature can be computed using the results of section \ref{Sec:CKinstantons}. The ordinary Eguchi-Hanson instanton corresponds to \eqref{SDGEH} with $*C_{abcd}=-C_{abcd}$, $b_{+}=0$ and $a_{+}=-a/16$ (the Ricci tensor then vanishes and the space is hyper-K\"ahler). The case \eqref{SDGEH} with $*C_{abcd}=-C_{abcd}$ and $b_{+}\neq0$ was studied in \cite{Giribet} in the context of conformal gravity, where the term $4b_{+}/r^2$ is referred to as the $b$-mode.

\paragraph{Generalised Taub-NUT.}
Letting $F=\frac{1}{f(r)}$, $G=4n^2f(r)$, $H=r^2-n^2$, where $f(r)$ is an arbitrary function of $r$ and $n$ is a constant,  the metric \eqref{PPclass} is a ``generalised Taub-NUT'' space. Putting $c_{\pm}=1$ in \eqref{CKPP}, we find $\Omega_{\pm}^{2}=(r\pm n)^{-2}$. The algebraic equation \eqref{PPSDsol} now gives:
\begin{align}
 *C_{abcd} = \mp C_{abcd} \quad \iff \quad f(r) = \frac{1}{4n^2}\left[ r^2 - n^2 +b_{\pm}(r\pm n)^{2}+a_{\pm}\frac{(r\pm n)^{3}}{(r\mp n)} \right].
\end{align}
Using \eqref{candidatesKilling}, we get $\xi^{a}_{\pm}\partial_{a}=\frac{1}{2n}\partial_{\psi}$, so $\xi^{a}_{+}=\xi^{a}_{-}$ is Killing and we can compute the rest of the curvature using the results of section \ref{Sec:CKinstantons}.

\subsection{Cosmological Einstein-Maxwell solutions}
\label{Sec:EMPP}

For concreteness, let us focus on the ASD side $\kappa^{-}$. Introducing new variables $(z,W,u)$ by
\begin{align}
\d{z}=\sqrt{FG}\d{r}, \qquad W=G^{-1}, \qquad e^{u} = GHh, \label{TodaVariablesPP}
\end{align}
the metric \eqref{PPclass} and fundamental 2-form $\kappa^{-}$ \eqref{ffPP} adopt the form \eqref{metric-gen}-\eqref{F2form-gen}. From \eqref{candidatesKilling}, we get $\xi^{a}_{-}\partial_{a}=\frac{\d(\Omega^{-1}_{-})}{\d{z}}\partial_{\psi}$. To apply the construction of section \ref{Sec:CKinstantons}, we need to restrict to the case in which $\xi^{a}_{-}$ is Killing. This is true iff $\d(\Omega_{-}^{-1})/\d{z}$ is a constant.   Given that $z$ and $\Omega_{-}$ are defined up to addition and multiplication by a constant respectively, we can then simply set $\Omega_{-}^{-1}\equiv z$. From the conformally K\"ahler condition $\Omega_{-}^{2}+\frac{\d}{\d{z}}(\Omega^2_{-}H)=0$ it follows that $\frac{\d}{\d{z}}(\frac{H}{z^2})=-\frac{1}{z^2}$, so we deduce
\begin{align}
 H= z + k z^2, \label{HEML}
\end{align}
where $k$ is an arbitrary constant. Now, from \eqref{CKPP} we have $\Omega^{2}_{+} = \frac{c_{+}c_{-}}{(H\Omega_{-})^2}$. Setting $c_{+}c_{-}=k$ for later convenience, we deduce $\Omega_{+}^{-1}=\frac{(1+kz)}{k}$. This implies $\d(\Omega^{-1}_{+})/\d{z}=1$, so $\xi^{a}_{+}\equiv\xi^{a}_{-}$.

The cosmological Einstein-Maxwell equations are $R=4\lambda$, where $R$ is given by \eqref{Todag}. Using that formula and the definitions \eqref{TodaVariablesPP}, we find
\begin{align}
 R = \frac{1}{H}\left[ R_{\Sigma} - \frac{\d^2(GH)}{\d{z}^{2}} \right] \label{RPPEM}
\end{align}
where $R_{\Sigma}$ was defined in \eqref{RPP}. It is convenient to have a formulation that is more symmetric in the SD and ASD sides. To this end, introduce $\varrho$ by $z = \frac{1}{\sqrt{k}}(\varrho-n)$,  where $n:=\frac{1}{2\sqrt{k}}$. Then $H=\varrho^2-n^2$ and  $\Omega_{\pm}^{-1}=\frac{1}{\sqrt{k}}(\varrho\pm n)$. 
The equation $R=4\lambda$ can be easily solved: from \eqref{RPPEM}, we find that $R_{\Sigma}$ must be constant and 
\begin{align}
 kG = \frac{-\frac{\lambda}{3}\varrho^4+(\frac{R_{\Sigma}}{2}+2\lambda n^2)\varrho^2+\alpha\varrho+\beta}{\varrho^2-n^2}, \label{GPP}
\end{align}
where $\alpha,\beta$ are arbitrary constants of integration. The solution then depends on 5 parameters: $k$ (or $n$), $R_{\Sigma},\lambda,\alpha,\beta$. To interpret them, we compute the rest of the curvature.

Recalling formulas \eqref{Psi2AK} and \eqref{RPP}, and using $\hat{G}_{\pm}\hat{z}_{\pm}=-\Omega_{\pm}^{4}GH$, we have
\begin{align*}
 \Psi_{2}^{\pm}=\frac{1}{12(\varrho^2-n^2)}\left[R_{\Sigma}-(\varrho\pm n)^2\tfrac{\d}{\d\varrho}\left((\varrho\pm n)^2\tfrac{\d}{\d\varrho}\left( \tfrac{kGH}{(\varrho\pm n)^4}\right) \right)\right].
\end{align*}
Using \eqref{GPP}, we find
\begin{align}
 \Psi_{2}^{\pm}=-\frac{\frac{1}{2}(\alpha\mp(R_{\Sigma}n+\frac{8}{3}\lambda n^3))}{(\varrho\pm n)^3}
 - \frac{(\beta-\frac{R_{\Sigma}}{2}n^2-\lambda n^4)}{(\varrho\mp n)(\varrho\pm n)^{3}}.
 \label{Psi2PP}
\end{align}
The SD piece of the Maxwell field is $F^{+}=\frac{z^{2}}{4}(\rho-\lambda\kappa)$. Recalling \eqref{TraceFreeRF}, we get
\begin{align}
F^{+}=-\frac{1}{4k}\frac{(\beta-\frac{R_{\Sigma}}{2}n^2-\lambda n^4)}{(\varrho+n)^2}\left[(\d\psi+A)\wedge\tfrac{\d\varrho}{\sqrt{k}}-(\varrho^2-n^2)h\d{x}\wedge\d{y} \right].
\label{SDMPP}
\end{align}
Formulas \eqref{Psi2PP}-\eqref{SDMPP} suggest to define
\begin{align}
 Q:=\beta-\tfrac{R_{\Sigma}}{2}n^2-\lambda n^4, \qquad \mu:=-\tfrac{1}{2}\alpha, \qquad \nu:=\tfrac{1}{2}(R_{\Sigma}n+\tfrac{8}{3}\lambda n^3),
\end{align}
and to identify $Q$ with ``electromagnetic charge'', $\mu$ with ``mass'', and $\nu$ with a sort of ``NUT charge''. The geometry is Einstein iff $Q=0$, and it is conformally SD iff $\mu=\nu$ and $Q=0$. In the latter case, the space is actually quaternionic-K\"ahler. The hyper-K\"ahler case  corresponds to $Q=\mu-\nu=\lambda=0$ , and, assuming $\Sigma=\CP^1$ (so $R_{\Sigma}=2$), it reduces to the Taub-NUT instanton with NUT charge $\nu=n$.

\section{The Pleba\'nski-Demia\'nski class}
\label{Sec:PD}

\subsection{$SU(\infty)$ Toda formulation}

The Pleba\'nski-Demia\'nski ansatz \cite{PD} is the 4-dimensional family of Riemannian metrics given in local real coordinates $(\tau,\phi,p,q)$ by
\begin{align}
g = \frac{1}{(p-q)^2}\left[ -Q\frac{(\d\tau-p^2\d\phi)^2}{(1-p^2q^2)}+P\frac{(\d\phi - q^2\d\tau)^2}{(1-p^2q^2)}+(1-p^2q^2)\left(\frac{\d{p}^2}{P}-\frac{\d{q}^2}{Q} \right) \right], \label{PDmetric}
\end{align}
where $P$ and $Q$ are arbitrary functions of $p$ and $q$ respectively, and we assume $P>0$, $Q<0$. The vector fields $\partial_{\tau}, \partial_{\phi}$ are Killing. We will first show that, similarly to \eqref{PPclass}, the geometry \eqref{PDmetric} is ambi-K\"ahler regardless of the form of $P,Q$. 

Consider the following orthonormal coframe:
\begin{equation}
\begin{aligned}
 \beta^{0} ={}& \tfrac{1}{(p-q)}\sqrt{\tfrac{-Q}{1-p^2q^2}}(\d\tau-p^2\d\phi), &
\beta^{1} =& \tfrac{1}{(p-q)}\sqrt{\tfrac{1-p^2q^2}{-Q}}\d{q}, \\
 \beta^{2} ={}&  \tfrac{1}{(p-q)}\sqrt{\tfrac{1-p^2q^2}{P}}\d{p}, &
 \beta^{3} =& \tfrac{1}{(p-q)}\sqrt{\tfrac{P}{1-p^2q^2}}(\d\phi - q^2\d\tau).
\end{aligned}
\end{equation}
Defining the 2-forms
\begin{align}
\kappa^{\pm}=\beta^0\wedge\beta^1\mp\beta^2\wedge\beta^3 = 
\frac{(\d\tau-p^2\d\phi)\wedge\d{q}\mp\d{p}\wedge(\d\phi-q^2\d\tau)}{(p-q)^2}
\end{align}
($\kappa^{+}$ is SD and $\kappa^{-}$ is ASD), the tensor fields $(J_{\pm})^{a}{}_{b}=\kappa^{\pm}_{bc}g^{ca}$ are almost-complex structures with opposite orientation. The type-$(1,0)$ eigenspace of $J_{\pm}$ is spanned by $\ell=\frac{1}{\sqrt{2}}(\beta^0+\i\beta^1)$, $m^{\pm}=\frac{1}{\sqrt{2}}(\beta^2\mp\i\beta^3)$. Setting
\begin{align}
 a^{0}_{\pm}:=\tfrac{(p-q)}{\sqrt{-Q(1-p^2q^2)}}, \quad 
 b^{0}_{\pm} := \pm \i \tfrac{(p-q)p^2}{\sqrt{P(1-p^2q^2)}}, \quad 
 a^{1}_{\pm}:= \i \tfrac{(p-q)q^2}{\sqrt{-Q(1-p^2q^2)}}, \quad 
 b^{1}_{\pm} := \mp \tfrac{(p-q)}{\sqrt{P(1-p^2q^2)}},
\end{align}
a straightforward calculation shows that the 1-forms $\omega^{\alpha}_{\pm}:=a^{\alpha}_{\pm}\ell+b^{\alpha}_{\pm}m^{\pm}$ (with $\alpha=0,1$) are 
\begin{align}
\omega^{0}_{\pm} = \frac{1}{\sqrt{2}}\left(\d\tau-\i\frac{\d{q}}{Q}\pm \i\frac{p^2\d{p}}{P}\right), \qquad 
\omega^{1}_{\pm} = \frac{1}{\sqrt{2}}\left(\i\d\phi+\frac{q^2\d{q}}{Q} \mp \frac{\d{p}}{P}\right), \label{10formsPD}
\end{align}
so $\d\omega^{\alpha}_{\pm}=0$. Since $\omega^{\alpha}_{\pm}$ are type-(1,0) forms for $J_{\pm}$, we see that both almost-complex structures $J_{\pm}$ are integrable\footnote{Note that from \eqref{10formsPD} we can also read off holomorphic coordinates for $J_{\pm}$, since $\omega^{\alpha}_{\pm}\equiv\d{z}^{\alpha}_{\pm}$.}. Additionally, we find 
\begin{align}
 \d[\Omega_{\pm}^{2}\kappa^{\pm}] = 0, \qquad \Omega_{\pm} = \frac{p-q}{1\pm pq},
\end{align}
thus, the class of metrics \eqref{PDmetric} is ambi-K\"ahler for any functions $P(p), Q(q)$. A computation shows that the vector fields defined by \eqref{CandKAK} are
\begin{align}
 \xi^{a}_{\pm}\partial_{a} = \partial_{\tau} \mp \partial_{\phi},
\end{align}
so both $\xi^{a}_{\pm}$ are Killing. We can then compute the Toda variables \eqref{TVAK}:
\begin{equation}
\begin{aligned}
& z_{\pm}=\frac{1\pm pq}{p-q}, \qquad W_{\pm}^{-1} = \frac{(1\pm q^2)^2P-(1\pm p^2)^2Q}{(p-q)^2(1-p^2q^2)}, \qquad e^{u_{\pm}} = \frac{-PQ}{(p-q)^4}, \\
& \d{x}_{\pm} = \frac{(1\pm p^2)}{P}\d{p}-\frac{(1\pm q^2)}{Q}\d{q}, \qquad y_{\pm} = -\tau \mp \phi.
\label{TVPD}
\end{aligned}
\end{equation}
Note that $\partial_{y_{\pm}}$ are Killing. We will now use this formulation to study field equations.

\subsection{A conformally self-dual Pleba\'nski-Demia\'nski space}

\begin{theorem}\label{Result:SDPD}
The metric \eqref{PDmetric} satisfies the conformally SD equation $*C_{abcd}=C_{abcd}$ if and only if the functions $P$ and $Q$ are given by
\begin{equation}\label{PQSD}
\begin{aligned}
 P ={}& a_{0}+a_{1}p+a_{2}p^{2}+a_{3}p^{3}+a_{4}p^{4}, \\
 Q={}& a_{4}+a_{3}q+a_{2}q^{2}+a_{1}q^{3}+a_{0}q^{4},
\end{aligned}
\end{equation}
where $a_{0},...,a_{4}$ are arbitrary constants. The solution is conformally half-flat but generically non-Einstein. Furthermore, the space is:
\begin{enumerate}
\item[$(i)$] Quaternionic-K\"ahler (i.e. Einstein) iff $a_1=a_3$, 
\item[$(ii)$] Hyper-K\"ahler (i.e. Ricci-flat) iff $a_{1}=a_{3}$ and $a_{0}=a_{4}$,
\item[$(iii)$] Flat iff $a_{1}=a_{3}=0$ and $a_{0}=a_{4}$.
\end{enumerate}
\end{theorem}

\begin{remark}
We stress that the conformally SD solution \eqref{PQSD} is different from the self-dual limit of the standard Pleba\'nski-Demia\'nski solution, which is a quaternionic-K\"ahler space corresponding to case $(i)$ above (see the next subsection). The solution \eqref{PQSD} can be regarded (locally) as a generalisation of the standard Pleba\'nski-Demia\'nski space to a self-dual gravitational instanton in conformal gravity.
\end{remark}

\begin{proof}[Proof of Theorem \ref{Result:SDPD}]
Recall that the conformally SD equation $*C_{abcd} = C_{abcd}$ is equivalent to $\hat{R}_{-}=0$, where $\hat{R}_{-}$ is given by \eqref{Todaghat}. For notational convenience, let us denote $x\equiv x_{-}$, $y\equiv y_{-}$, $z\equiv z_{-}$, $u\equiv u_{-}$. Since $\partial_{y}$ is Killing, the ansatz \eqref{PDmetric} satisfies $*C_{abcd} = C_{abcd}$ if and only if
\begin{align}
\hat{u}_{xx} + (e^{\hat{u}})_{\hat{z}\hat{z}}=0, \label{HatTodaPD}
\end{align}
where $\hat{u}=u-4\log{z}$ and $\hat{z}=-\frac{1}{z}$. If one writes the Toda equation \eqref{HatTodaPD} in terms of the variables $p,q,P,Q$ (using \eqref{TVPD} for the $-$ sign) and tries to solve for $P,Q$ by brute force, the equation becomes too complicated and we were not able to solve it in this way. Instead, in order to solve \eqref{HatTodaPD} we recall the trick \eqref{trickKN} that we used to solve the Toda equation in the Kerr-Newman case (section \ref{Sec:Examples}): we introduce an auxiliary variable $\hat\sigma$ by 
\begin{align}
 \hat{u}_{x} = \hat\sigma_{\hat{z}}, \qquad (e^{\hat{u}})_{\hat{z}} = - \hat\sigma_{x}. \label{trickPD}
\end{align}
The vector fields $\partial_{x},\partial_{z}$ can be computed from \eqref{TVPD}: we find
\begin{align}
\partial_{x}=\frac{PQ}{F}\left[(1-p^2)\partial_{p}+(1-q^2)\partial_{q} \right], \qquad 
\partial_{z}=\frac{(p-q)^2}{F}\left[(1-q^2)P\partial_{p}+(1-p^2)Q\partial_{q} \right],
\label{VFPD}
\end{align}
where $F\equiv(1-p^2)^2Q-(1-q^2)^2P$. Noticing that $\partial_{\hat{z}}=z^{2}\partial_{z}$, eqs. \eqref{trickPD} lead, respectively, to
\begin{subequations}
\begin{align}
\frac{(1-p^2)Q\dot{P}}{(1-pq)^2}+\frac{(1-q^2)P\dot{Q}}{(1-pq)^2}+\frac{4(p+q)PQ}{(1-pq)^2} ={}& 
(1-q^2)P\frac{\partial\hat\sigma}{\partial{p}} + (1-p^2)Q\frac{\partial\hat\sigma}{\partial{q}}, \label{hatxzPD} \\
\frac{(1-q^2)\dot{P}}{(1-pq)^2}+\frac{(1-p^2)\dot{Q}}{(1-pq)^2}+\frac{4q(1-q^2)P}{(1-pq)^3}+\frac{4p(1-p^2)Q}{(1-pq)^3} ={}& (1-p^2)\frac{\partial\hat\sigma}{\partial{p}} + (1-q^2)\frac{\partial\hat\sigma}{\partial{q}}, \label{hatzxPD}
\end{align}
\end{subequations}
where $\dot{P}=\frac{\d{P}}{\d{p}}$, $\dot{Q}=\frac{\d{Q}}{\d{q}}$. Now, from \eqref{hatzxPD} we find an expression for $\partial_{p}\hat\sigma$, and we then replace this in \eqref{hatxzPD}. When we do this, $\dot{Q}$ disappears from the resulting equation, leaving us with an equation for $\partial_{q}\hat\sigma$ and $\dot{P}$ only. We then replace this new expression for $\partial_{q}\hat\sigma$ in \eqref{hatzxPD}, and we end up with an equation for $\partial_{p}\hat\sigma$ and $\dot{Q}$ only. Explicitly, we find:
\begin{align*}
 \frac{\partial\hat\sigma}{\partial{q}} = \frac{\dot{P}}{(1-pq)^2} + \frac{4qP}{(1-pq)^3}, \qquad
 \frac{\partial\hat\sigma}{\partial{p}} = \frac{\dot{Q}}{(1-pq)^2} + \frac{4pQ}{(1-pq)^3}.
\end{align*}
Using these equations and the identity $\partial_{p}\partial_{q}\hat\sigma=\partial_{q}\partial_{p}\hat\sigma$, a short calculation leads to 
\begin{align}
(1-pq)^{2}\ddot{Q}+6p(1-pq)\dot{Q}+12p^2Q=(1-pq)^{2}\ddot{P}+6q(1-pq)\dot{P}+12q^2P.
\label{ConstSDPD}
\end{align}
Applying $\partial^{2}_{q}$ to this equation, and then $\partial^{2}_{p}$ to the resulting expression, we get 
\begin{align*}
 q^2\ddddot{Q}-2q\dddot{Q}+2\ddot{Q}=p^2\ddddot{P}-2p\dddot{P}+2\ddot{P},
\end{align*}
which can be rewritten as 
\begin{align*}
q^{3}\frac{\d^2}{\d{q}^2}\left(\frac{\ddot{Q}}{q}\right)=p^{3}\frac{\d^2}{\d{p}^2}\left(\frac{\ddot{P}}{p}\right).
\end{align*}
Since the left side is a function of $q$ only, and the right side is a function of $p$ only, the equation is easy to solve: we find that $P$ and $Q$ must be fourth order polynomials in $p$ and $q$, respectively. In addition, the fact that $P$ and $Q$ must satisfy \eqref{ConstSDPD} imposes relations between the coefficients of the polynomials: this then leads to the form \eqref{PQSD}.

It remains to prove the assertion concerning the special limits $(i),(ii),(iii)$. We find
\begin{align}
\frac{W^{-}_{0}}{W_{-}} = z_{-}^{3}\left[ (a_4-a_0)+\frac{(a_3-a_1)}{2}\frac{(p+q)}{(1+pq)} \right].
\label{W0W}
\end{align}
Now we use Theorem \ref{Result:Ricci}, from where we see that the solution will be Einstein iff $\frac{1}{z^{2}_{-}}\partial_{z_{-}}(\frac{W^{-}_{0}}{W_{-}})=\frac{R}{4}=\lambda$. Since it is conformally self-dual, the Einstein condition will imply that it is quaternionic-K\"ahler. From \eqref{W0W} we see that this is true iff $a_{1}=a_{3}$. The cosmological constant is $\lambda=3(a_4-a_0)$, and the only non-vanishing component of the SD Weyl spinor is
\begin{align}
\Psi^{+}_{2} = -\frac{a_{1}}{z^{3}_{+}}. \label{Psi2+PD}
\end{align}
In addition, the solution will be hyper-K\"ahler iff $R_{ab}=0$, which from the above reduces to $a_1=a_3$ and $a_0=a_{4}$. The only non-trivial part of the curvature is now \eqref{Psi2+PD}. Finally, from these considerations and eq. \eqref{Psi2+PD} we see that the solution will be flat iff $a_1=a_3=0$ and $a_0=a_{4}$. 
\end{proof}

Note that in the flat limit there are still two parameters left ($a_0$ and $a_2$), so we actually get a 2-parameter family of flat metrics, as is expected from the analysis in \cite{PD}.

\subsection{Cosmological Einstein-Maxwell solutions}

Although the Pleba\'nski-Demia\'nski solution to the system \eqref{CEM} is well-known \cite{PD}, here we re-derive the result as an application of the framework developed in section \ref{Sec:CKinstantons}. This illustrates that one actually does not need to solve the full Einstein equations as in \cite{PD}, but just $R=4\lambda$. This example also allows us to give a trick to solve the modified Toda equation.

\begin{proposition}\label{Result:CEMPD}
The metric \eqref{PDmetric} satisfies the cosmological Einstein-Maxwell equations \eqref{CEM} if and only if the functions $P$ and $Q$ are given by
\begin{equation}\label{CEMPD}
\begin{aligned}
 P ={}& a_{0}+a_{1}p+a_{2}p^{2}+a_{3}p^{3}+a_{4}p^{4}, \\
 Q={}& (a_{0}+\tfrac{1}{3}\lambda)+a_{1}q+a_{2}q^{2}+a_{3}q^{3}+(a_{4}-\tfrac{1}{3}\lambda)q^{4},
\end{aligned}
\end{equation}
where $a_{0},...,a_{4}$ are arbitrary constants. 
\end{proposition}

\begin{proof}
For concreteness, we choose to work with the ASD side, and we denote $x\equiv x_{-}$, $y\equiv y_{-}$, $z\equiv z_{-}$, $u\equiv u_{-}$, $W\equiv W_{-}$. 
Since the ansatz \eqref{PDmetric} is conformally K\"ahler with symmetry, the cosmological Einstein-Maxwell equations reduce to
\begin{align}
 u_{xx}+(e^{u})_{zz}=-4\lambda We^{u} \label{MTPD}
\end{align}
(as $\partial_{y}$ is Kiling). To solve the modified Toda equation \eqref{MTPD}, we use a slight variation of the trick used in \eqref{trickPD}: we introduce two variables $\sigma,T$ by
\begin{align}
 u_{x} = \sigma_{z} + T, \qquad (e^{u})_{z} = -\sigma_{x}. \label{trick2PD}
\end{align}
Equation \eqref{MTPD} becomes $T_{x}=-4\lambda We^{u}$, and, using \eqref{VFPD}, this gives
\begin{align}
(1-p^2)\partial_{p}T+(1-q^2)\partial_{q}T = -4\lambda\frac{(1-p^2q^2)}{(p-q)^2}. \label{EqForT}
\end{align}
Equations \eqref{trick2PD} lead to
\begin{align*}
\tfrac{(1-p^2)}{(p-q)^2}Q\dot{P}+\tfrac{(1-q^2)}{(p-q)^2}P\dot{Q}+\tfrac{4(p+q)}{(p-q)^2}PQ ={}& (1-q^2)P \partial_{p}\sigma + (1-p^2)Q \partial_{q}\sigma+\tfrac{F}{(p-q)^2}T, \\
\tfrac{(1-q^2)}{(p-q)^2}\dot{P}+\tfrac{(1-p^2)}{(p-q)^2}\dot{Q}-\tfrac{4(1-q^2)}{(p-q)^3}P+\tfrac{4(1-p^2)}{(p-q)^3}Q ={}& (1-p^2)\partial_{p}\sigma + (1-q^2)\partial_{q}\sigma,
\end{align*}
where $F=(1-p^2)^2Q-(1-q^2)^2P$. Proceeding as in the proof of Theorem \ref{Result:SDPD}, we now arrive at the system
\begin{align*}
 \frac{\partial\sigma}{\partial{q}} = \frac{\dot{P}}{(p-q)^2} - \frac{4P}{(p-q)^3} - \frac{(1-p^2)}{(p-q)^2}T, \qquad
 \frac{\partial\sigma}{\partial{p}} = \frac{\dot{Q}}{(p-q)^2} + \frac{4Q}{(p-q)^3}+\frac{(1-q^2)}{(p-q)^2}T.
\end{align*}
Using the identity $\partial_{p}\partial_{q}\sigma=\partial_{q}\partial_{p}\sigma$ and eq. \eqref{EqForT}, we get
\begin{align}
(p-q)^{2}\ddot{Q}+6(p-q)\dot{Q}+12Q=(p-q)^{2}\ddot{P}-6(p-q)\dot{P}+12P +4\lambda(1-p^2q^2).
\label{ConstEMPD}
\end{align}
Applying $\partial^{2}_{q}$ and then $\partial^{2}_{p}$ we are led to
\begin{align}
 \ddddot{Q} - \ddddot{P} = -8\lambda.
\end{align}
Taking additional derivatives $\partial_{p}$ and $\partial_{q}$, and using that $P$ and $Q$ depend only on $p$ and $q$ respectively, we see that $P$ and $Q$ must be fourth order polynomials, $P=\sum_{i=0}^{4}a_ip^i$, $Q=\sum_{i}^{4}b_iq^i$. Replacing back in \eqref{ConstEMPD}, we get $b_0=a_0+\frac{1}{3}\lambda$, $b_1=a_1$, $b_2=a_2$, $b_3=a_3$, $b_4=a_4-\frac{1}{3}\lambda$, so the result \eqref{CEMPD} follows.
\end{proof}

Using \eqref{Psi2AK}, we find: 
\begin{align}
 \frac{W_{0}^{\pm}}{W_{\pm}} ={}& \frac{(a_{3}\pm a_{1})}{2} - (a_0-a_4+\frac{1}{3}\lambda)\left(\frac{p+q}{1\mp pq}\right) + \frac{\lambda}{3}\left(\frac{1\pm pq}{p-q}\right)^{3}, \\
 \Psi^{\pm}_{2} ={}& -\frac{(a_3\pm a_1)}{2}\left(\frac{p-q}{1\pm pq}\right)^{3} + 
 (a_0-a_4+\frac{1}{3}\lambda)\left(\frac{p+q}{1\mp pq}\right)\left(\frac{p-q}{1\pm pq}\right)^{3}.
\end{align}
From these formulas we see that the conformally SD limit $\Psi_{2}^{-}=0$ corresponds to $a_{3}=a_{1}$ and $a_{0}-a_{4}+\frac{\lambda}{3}=0$, which implies $\frac{W_{0}^{\pm}}{W_{\pm}}=\frac{\lambda}{3}z_{\pm}^{3}$. Using then Theorem \ref{Result:Ricci}, in this limit we get $\rho=\lambda\kappa$, so the space is Einstein. Thus, the conformally SD limit of the standard Pleba\'nski-Demia\'nski solution \eqref{CEMPD} is indeed different from the generalisation found in Theorem \ref{Result:SDPD}.

\section{The Chen-Teo class}
\label{Sec:ChenTeo}

\subsection{$SU(\infty)$ Toda formulation}

Consider the 4-dimensional family of metrics given in local coordinates $(\tau,\phi,x_{1},x_{2})$ by 
\begin{align}
 g = \frac{(F\d\tau+G\d\phi)^2}{(x_1-x_2)HF} + \frac{kH}{(x_1-x_2)^3}
 \left(\frac{\d{x}_1^2}{X_1}-\frac{\d{x}_2^2}{X_2}-\frac{X_1X_2}{kF}\d\phi^2 \right), 
 \label{ChenTeo}
\end{align}
where $k$ is a constant, $G(x_1,x_2),H(x_1,x_2),X_1(x_1),X_2(x_2)$ are arbitrary functions of their arguments, and 
\begin{align}
 F = x_2^2 X_1 - x_1^2 X_2.
\end{align}
The vector fields $\partial_{\tau},\partial_{\phi}$ are Killing. 
For a specific choice of the functions $G,H,X_1,X_2$, the metric \eqref{ChenTeo} is the Ricci-flat Chen-Teo geometry, see \cite[Eq. (2.1)]{ChenTeo2} \footnote{To compare our notation to that of \cite{ChenTeo2}, set $X_1\equiv X$, $X_2\equiv Y$, $x_1\equiv x$, $x_2\equiv y$.}.

Let $c$ be an arbitrary constant, and define new variables 
\begin{equation}
\begin{aligned}
& W:=\frac{k}{c^2}\frac{(x_1-x_2)H}{F}, \qquad \psi:=\frac{\sqrt{k}}{c}\tau, \qquad y:=\frac{c}{\sqrt{k}}\phi, 
\qquad \tilde{G}:=\frac{k}{c^2}\frac{G}{F}, \qquad A:=\tilde{G}\d{y} \\
& \d{x} := c\left[\frac{x_1}{X_1}\d{x}_1-\frac{x_2}{X_2}\d{x}_2 \right], \qquad 
\d{z} := c\frac{(x_2\d{x}_1-x_1\d{x_2})}{(x_1-x_2)^2}, \qquad e^{u}:=\frac{-X_1X_2}{(x_1-x_2)^4}.
\label{TVCT}
\end{aligned}
\end{equation}
Then a calculation shows that \eqref{ChenTeo} becomes 
\begin{align}
 g = W^{-1}(\d\psi+A)^2+W[\d{z}^2+e^{u}(\d{x}^2+\d{y}^2)]. \label{gTCT}
\end{align}
The Killing fields are now $\partial_{\psi},\partial_{y}$.

\begin{remark}[The Chen-Teo parameter $\nu$]
\label{Remark:nu}
From the expression for $\d{z}$ in \eqref{TVCT} we can find $z$ by integration: the solution is $z = \frac{cx_2}{x_2-x_1} + \nu$, where $\nu$ is an arbitrary constant. We are free to choose any relation between $c$ and $\nu$ we want; in particular, setting
\begin{align}
 c\equiv -(1+\nu), \label{nuCT}
\end{align}
we get $z = \frac{\nu x_{1}+x_{2}}{x_{1}-x_{2}}$, which, in the Ricci-flat Chen-Teo case, is the (inverse of the) conformal factor that makes the metric K\"ahler. The parameter $\nu$ is particularly important in the Ricci-flat case \cite{ChenTeo2}: the Chen-Teo solution is a one-parameter ($-1\leq\nu\leq1$) family of metrics interpolating between the Pleba\'nski-Demia\'nski ($\nu=1$) and Gibbons-Hawking ($\nu=-1$) spaces.
\end{remark}

The fact that the metric \eqref{ChenTeo} can be written as \eqref{gTCT} does not imply, of course, that the geometry \eqref{ChenTeo} is necessarily conformally K\"ahler.  To investigate this, we choose the coframe $\beta^0=W^{-1/2}(\d\psi+A)$, $\beta^1=\sqrt{W}\d{z}$, $\beta^2=\sqrt{W}e^{u/2}\d{x}$, $\beta^{3}=\sqrt{W}e^{u/2}\d{y}$ for \eqref{gTCT}. The 2-form $\kappa=\beta^0\wedge\beta^1+\beta^2\wedge\beta^3$ is equal to \eqref{F2form-gen}, and defines the almost-complex structure $J^{a}{}_{b}=\kappa_{bc}g^{ca}$. The type-$(1,0)$ eigenspace is spanned by $\ell=\frac{1}{\sqrt{2}}(\beta^0+\i\beta^1)$, $m=\frac{1}{\sqrt{2}}(\beta^2+\i\beta^3)$. In particular, the following are type-(1,0) forms:
\begin{align}
 \omega^{0}=\d\psi+A+\i W\d{z}+B(\d{x}+\i\d{y}), \qquad \omega^{1}=\d{x} + \i\d{y}
 \label{JCT}
\end{align}
where $B$ is an arbitrary complex function. Since $\d\omega^1=0$, we see that $J$ will be integrable if $\d\omega^0=0$. This gives the Hermitian condition, and, assuming that it holds, the conformally K\"ahler condition is $\d(z^{-2}\kappa)=0$. These two conditions lead respectively to:
\begin{subequations}\label{condCKCT}
\begin{align}
 Z_{1} :={}& \tilde{G}_{z} - W_{x}= 0, \label{HermitianCT} \\
 Z_{2} :={}& \tilde{G}_{x}+z^2\partial_{z}\left( \tfrac{We^u}{z^2} \right)=0 \label{CKCT}
\end{align}
\end{subequations} 
(recall \eqref{monopoleEq}).
The geometry \eqref{ChenTeo} will be conformally K\"ahler (for the given choice of almost-complex structure \eqref{JCT}) iff the conditions \eqref{condCKCT} are satisfied. The vector field \eqref{vectorCKY} is the Killing vector $\partial_{\psi}$. We also find
\begin{equation}
\begin{aligned}
\partial_{x} ={}& -\frac{X_{1}X_{2}}{cF}(x_{1}\partial_{x_{1}}+x_{2}\partial_{x_{2}}), \\ 
\partial_{z} ={}& \frac{(x_{1}-x_{2})^2}{cF}(x_{2}X_{1}\partial_{x_{1}}+x_{1}X_{2}\partial_{x_{2}}). 
\label{VFCT}
\end{aligned}
\end{equation}

\begin{remark}
To have some intuition about \eqref{condCKCT}, we express $Z_1$ in terms of the original variables:
\begin{align}
Z_{1}=\frac{k(x_1-x_2)X_1X_2}{c^2F}\left[ \partial_{x_2}\left( \frac{x_1G}{X_1F}+\frac{x_2H}{(x_1-x_2)F}\right) + \partial_{x_1}\left( \frac{x_2G}{X_2F}+\frac{x_1H}{(x_1-x_2)F} \right) \right].
\end{align}
Then, for the original Chen-Teo Ricci-flat metric \cite{ChenTeo2}, using \cite[Eqs. (3.7a)-(3.7b)]{AkAnd} we see that indeed $Z_1=0$, which justifies our choice of almost-complex structure \eqref{JCT} for the general class \eqref{ChenTeo}. (In the Ricci-flat case, $Z_2=0$ follows form $Z_1=0$.)
\end{remark}

\subsection{A family of conformally self-dual Chen-Teo metrics}

\begin{theorem}\label{Result:SDCT}
Assume the family of metrics \eqref{ChenTeo} to be conformally K\"ahler w.r.t the almost-complex structure \eqref{JCT} (that is, conditions \eqref{condCKCT} are satisfied). Choose the relation \eqref{nuCT} between the parameters $c$ and $\nu$. Then \eqref{ChenTeo} is conformally self-dual $*C_{abcd} = C_{abcd}$ if and only if the functions $X_1,X_2$ are given by
\begin{equation}\label{SDCT}
\begin{aligned}
 X_1 ={}& a_0+a_1x_1+a_2x_1^2+a_3x_1^3+a_4x_1^4, \\
 X_2 ={}& a_0\nu^2-a_1\nu x_2 +a_2x_2^2-\tfrac{a_3}{\nu}x_2^3+\tfrac{a_4}{\nu^2}x_2^4.
\end{aligned}
\end{equation}
\end{theorem}

\begin{proof}
The proof is very similar to the proof of Theorem \ref{Result:SDPD}. We use that $*C_{abcd}=C_{abcd}$ is equivalent to $\hat{R}=0$, where $\hat{R}$ is given by \eqref{Todaghat}: $\hat{u}_{xx}+(e^{\hat{u}})_{\hat{z}\hat{z}}=0$, with $\hat{u}=u-4\log z$, $\hat{z}=-1/z$ (we used that $\partial_{y}$ is Killing). Introducing $\hat\sigma$ by $\hat{u}_x=\hat{\sigma}_{\hat{z}}$, $(e^{\hat{u}})_{\hat{z}}=-\hat{\sigma}_x$, and using that the vector fields $\partial_{x},\partial_{z}$ are given by \eqref{VFCT}, we are led to 
\begin{align*}
\frac{x_{1}X_{2}\dot{X}_{1}}{(\nu x_{1}+x_{2})^{2}} + \frac{x_{2}X_{1}\dot{X}_{2}}{(\nu x_{1}+x_{2})^{2}} - \frac{4X_{1}X_{2} }{(\nu x_{1}+x_{2})^{2}}
={}& -x_{2}X_{1}\frac{\partial\hat{\sigma}}{\partial{x}_{1}} - x_{1}X_{2}\frac{\partial\hat\sigma}{\partial{x}_{2}}, \\
\frac{x_{2}\dot{X}_{1}}{(\nu x_{1}+x_{2})^{2}}+\frac{x_{1}\dot{X}_{2}}{(\nu x_{1}+x_{2})^{2}} 
- \frac{4\nu x_{2}X_{1}}{(\nu x_{1}+x_{2})^{3}} - \frac{4 x_{1} X_{2}}{(\nu x_{1}+x_{2})^{3}}
={}& -x_{1}\frac{\partial\hat{\sigma}}{\partial{x}_{1}} - x_{2}\frac{\partial\hat\sigma}{\partial{x}_{2}},
\end{align*}
where $\dot{X}_{1}\equiv\frac{\d{X}_{1}}{\d{x}_{1}}$, $\dot{X}_{2}\equiv\frac{\d{X}_{2}}{\d{x}_{2}}$.
We then deduce
\begin{align*}
\frac{\partial\hat{\sigma}}{\partial{x}_{2}} = -\frac{\dot{X}_{1}}{(\nu x_{1}+x_{2})^{2}} + \frac{4\nu X_{1}}{(\nu x_{1}+x_{2})^{3}}, \qquad
\frac{\partial\hat{\sigma}}{\partial{x}_{1}} = -\frac{\dot{X}_{2}}{(\nu x_{1}+x_{2})^{2}} + \frac{4 X_{1}}{(\nu x_{1}+x_{2})^{3}}.
\end{align*}
Using $\partial_{x_{1}}\partial_{x_{2}}\hat{\sigma}=\partial_{x_{2}}\partial_{x_{1}}\hat{\sigma}$, we find
\begin{align}
(\nu x_{1}+x_{2})^{2}\ddot{X}_{1}-6(\nu x_{1}+x_{2})\dot{X}_{1}+12\nu^2 X_{1} 
= (\nu x_{1}+x_{2})^{2}\ddot{X}_{2}-6(\nu x_{1}+x_{2})\dot{X}_{2}+12 X_{2}. \label{ConstSDCT}
\end{align}
Applying $\partial^{2}_{x_{2}}\partial^{2}_{x_{1}}$, we get 
\begin{align}
 \ddddot{X}_{1} = \nu^{2}\ddddot{X}_{2},
\end{align}
which implies that $X_{1}$ and $X_{2}$ are fourth order polynomials in $x_{1}$ and $x_{2}$ respectively. Replacing in \eqref{ConstSDCT}, we find relations between the coefficients and we get \eqref{SDCT}.
\end{proof}

Similarly to Theorem \ref{Result:SDPD}, the solution \eqref{SDCT} can be regarded (locally) as a generalisation of the Ricci-flat Chen-Teo metric to a self-dual gravitational instanton in conformal gravity. However, unlike \eqref{PQSD}, the conformally self-dual equation now determines $X_1,X_2$ in \eqref{ChenTeo} to be given by \eqref{SDCT}, but it does not determine the other arbitrary functions $G,H$ in \eqref{ChenTeo}. These are constrained by the conformally K\"ahler condition \eqref{condCKCT}, but this restriction does not determine $G,H$ uniquely. 
This situation is analogous to what we encountered in Theorem \ref{Result:SDPP}.

\subsection{A family of Einstein-Maxwell Chen-Teo metrics}

\begin{theorem}\label{Result:EMCT}
Assume the family of metrics \eqref{ChenTeo} to be conformally K\"ahler w.r.t the almost-complex structure \eqref{JCT} (that is, conditions \eqref{condCKCT} are satisfied). Choose the relation \eqref{nuCT} between the parameters $c$ and $\nu$. Then \eqref{ChenTeo} satisfies the Einstein-Maxwell equations (with $\lambda=0$) if and only if 
\begin{equation}\label{EMCT}
\begin{aligned}
 X_1 ={}& a_0 + a_1 x_1 + a_2 x_1^2 + a_3 x_1^3 + a_4 x_1^4, \\
 X_2 ={}& a_0 + a_1 x_2 + a_2 x_2^2 + a_3 x_2^3 + a_4 x_2^4,
\end{aligned}
\end{equation}
where $a_0,...a_4$ are arbitrary constants.
\end{theorem}

\begin{proof}
Recall that the Einstein-Maxwell equations with $\lambda=0$ reduce to $u_{xx}+(e^{u})_{zz}=0$. The procedure to solve this is the same that we used in previous cases: we introduce an auxiliary variable $\sigma$ by $u_x=\sigma_z$, $(e^u)_z=-\sigma_x$. 
Using \eqref{VFCT}, we deduce 
\begin{align*}
& \frac{x_1X_{2}\dot{X}_{1}}{(x_1-x_2)^2}+\frac{x_2X_1\dot{X}_2}{(x_1-x_2)^2}-\frac{4X_1X_2}{(x_1-x_2)^2} = - x_2X_1\frac{\partial\sigma}{\partial{x}_1} - x_1X_2\frac{\partial\sigma}{\partial{x}_2}, \\
& \frac{x_2\dot{X}_{1}}{(x_1-x_2)^2}+\frac{x_1\dot{X}_2}{(x_1-x_2)^2}-\frac{4x_2X_1}{(x_1-x_2)^3}+\frac{4x_1X_2}{(x_1-x_2)^3} = - x_1\frac{\partial\sigma}{\partial{x}_1} - x_2\frac{\partial\sigma}{\partial{x}_2}.
\end{align*}
This leads to
\begin{align*}
 \frac{\partial\sigma}{\partial{x}_2} = -\frac{\dot{X}_1}{(x_1-x_2)^2}+\frac{4X_1}{(x_1-x_2)^3}, \qquad
 \frac{\partial\sigma}{\partial{x}_1} = -\frac{\dot{X}_2}{(x_1-x_2)^2}-\frac{4X_2}{(x_1-x_2)^3}.
\end{align*}
Using then $\partial_{x_1}\partial_{x_2}\sigma=\partial_{x_2}\partial_{x_1}\sigma$, after some calculations we arrive at
\begin{align}
(x_1-x_2)^{2}\ddot{X}_2+6(x_1-x_2)\dot{X}_2+12X_2 = (x_1-x_2)^{2}\ddot{X}_1-6(x_1-x_2)\dot{X}_2+12X_1. 
\label{EqForCT}
\end{align}
Applying $\partial_{x_1}^{2}\partial_{x_2}^{2}$, we get $\ddddot{X}_{1}=\ddddot{X}_{2}$, which implies that $X_{1},X_{2}$ are fourth order polynomials, and replacing in \eqref{EqForCT} we see that they must have the same coefficients.
\end{proof}

The previous result shows that we now encounter a new phenomenon for Einstein-Maxwell instantons, not present in the Pleba\'nski-Demia\'nski and previous cases: by construction, the conformal K\"ahler assumption and the fact that \eqref{EMCT} solve the $SU(\infty)$ Toda equation guarantee that the ansatz \eqref{ChenTeo} is automatically a solution to the Einstein-Maxwell system, with the Maxwell field given by \eqref{SDMaxwell} and \eqref{TraceFreeRF} (or \eqref{RicciFormFormula2}), but the functions $G$ and $H$ in \eqref{ChenTeo}, while constrained by \eqref{condCKCT}, remain undetermined. Since $u$ {\em is} determined, one can use that $H$ is related to $W$ by \eqref{TVCT} and that $W$ must satisfy 
\begin{align}
 W_{xx} + \partial_{z} \left[ z^2 \partial_{z}(\tfrac{We^u}{z^2}) \right] = 0 \label{EqForWCT}
\end{align}
(recall \eqref{EqForW}), to solve this equation for $H$, and then use this to find $G$ via \eqref{condCKCT}. However, the solution to this system is not unique. The situation can be compared to the Gibbons-Hawking ansatz, where instead of \eqref{EqForWCT} one has the Laplace equation and instead of \eqref{monopoleEq} one has the standard monopole equation $\d{A}=*_{3}\d{V}$ (see e.g. \cite{DunajskiBook}), but these equations do not have a unique solution. Our construction thus provides a {\em family} of Einstein-Maxwell Chen-Teo geometries, in a similar sense as the Gibbons-Hawking ansatz provides a family of hyper-K\"ahler geometries. 

By contrast, the Ricci-flat condition for the ansatz \eqref{ChenTeo} gives the same solution \eqref{EMCT} but it also determines $H$ and $G$, since $W$ must be given by \eqref{Wring}, and then one finds $H$ from \eqref{TVCT} and $G$ from \eqref{condCKCT}. We have checked that in this way one indeed recovers the original Ricci-flat Chen-Teo metric \cite[Eq. (2.1)]{ChenTeo2}. In the Einstein case with $\lambda\neq0$, $W$ is given by \eqref{Wlambda}, and one must solve the modified Toda equation. This can be done with the same trick that we used in the Pleba\'nski-Demia\'nski case (Theorem \ref{Result:CEMPD}); the result will be presented in a separate work.

\paragraph{Acknowledgements.} 
I am very grateful to Maciej Dunajski and Paul Tod for very helpful conversations about the topics of this work, and to the Institut Mittag-Leffler in Djursholm, Sweden for hospitality during the conference ``Einstein Spaces and Special Geometry'' in July 2023. I am also especially grateful to an anonymous Referee for many useful comments, suggestions of appropriate bibliographic references, and recommendations that led to a substantial improvement in the presentation of this work. Finally I would also like to thank the Alexander von Humboldt Foundation and the Max Planck Society for financial support during the completion of this work.

\paragraph{Conflict of interest.}
The author declares that there is no conflict of interest.


\begin{thebibliography}{99}


\bibitem{AkAnd}
S.~Aksteiner and L.~Andersson,
{\em Gravitational Instantons and special geometry},
\href{https://arxiv.org/abs/2112.11863}{[arXiv:2112.11863 [gr-qc]]}.

\bibitem{ApostolovGauduchon}
V.~Apostolov and P.~Gauduchon, 
{\em Selfdual Einstein hermitian four-manifolds}, 
Annali della Scuola Normale Superiore di Pisa-Classe di Scienze, 1(1), 203-243 (2002) 
\href{https://arxiv.org/abs/math/0003162}{[arXiv:math/0003162 [math.DG]]}

\bibitem{ACG}
V.~Apostolov, D.~Calderbank and P.~Gauduchon,  
{\em The Geometry of Weakly Self-dual K\"ahler Surfaces}, 
Compositio Mathematica, 135(3), 279-322 (2003) 
\href{https://arxiv.org/abs/math/0104233}{[arXiv:math/0104233 [math.DG]]}

\bibitem{ApostolovCalderbankGauduchon}
V.~Apostolov, D.~M.~J.~Calderbank and P.~Gauduchon,
{\em Ambitoric geometry I: Einstein metrics and extremal ambik\"ahler structures},
J. Reine Angew. Math. \textbf{2016} (2016) no.721, 109-147
\href{https://arxiv.org/abs/1302.6975}{[arXiv:1302.6975 [math.DG]]}.

\bibitem{ApostolovMaschler}
V.~Apostolov and G.~Maschler, 
{\em Conformally K\"ahler, Einstein-Maxwell geometry}, 
J. Eur. Math. Soc. 21 (2019), no. 5, pp. 1319-1360
\href{https://arxiv.org/abs/1512.06391}{[arXiv:1512.06391 [math.DG]]}

\bibitem{Besse}
A.~L.~Besse, 
{\em Einstein manifolds}, 
Classics in Mathematics. Springer-Verlag, Berlin, 2008. Reprint of the 1987 edition.

\bibitem{BG}
O.~Biquard and P.~Gauduchon,
{\em On Toric Hermitian ALF Gravitational Instantons},
Commun. Math. Phys. \textbf{399} (2023) no.1, 389-422
\href{https://arxiv.org/abs/2112.12711}{[arXiv:2112.12711 [math.DG]]}.

\bibitem{BG2}
O.~Biquard and P.~Gauduchon,  
{\em About a family of ALF instantons with conical singularities}, 
SIGMA. Symmetry, Integrability and Geometry: Methods and Applications, 19, 079 (2023) 
\href{https://arxiv.org/abs/2306.11110}{[arXiv:2306.11110 [math.DG]]}

\bibitem{Calderbank}
D.~M.~J.~Calderbank and H.~Pedersen,
{\em Selfdual Einstein metrics with torus symmetry},
J. Diff. Geom. \textbf{60} (2002) no.3, 485-521
\href{https://arxiv.org/abs/math/0105263}{[arXiv:math/0105263 [math.DG]].}


\bibitem{CC1}
G.~Chen and X.~Chen,
{\em Gravitational instantons with faster than quadratic curvature decay (I)},
Acta Math. \textbf{227} (2021), 263
\href{https://arxiv.org/abs/1505.01790}{[arXiv:1505.01790 [math.DG]].}

\bibitem{ChenTeo1}
Y.~Chen and E.~Teo,
{\em A New AF gravitational instanton},
Phys. Lett. B \textbf{703} (2011), 359-362
\href{https://arxiv.org/abs/1107.0763}{[arXiv:1107.0763 [gr-qc]]}.

\bibitem{ChenTeo2}
Y.~Chen and E.~Teo,
{\em Five-parameter class of solutions to the vacuum Einstein equations},
Phys. Rev. D \textbf{91} (2015) no.12, 124005
\href{https://arxiv.org/abs/1504.01235}{[arXiv:1504.01235 [gr-qc]]}.

\bibitem{CVZ}
G. Chen, J. Viaclovasky, and R. Zhang, 
{\em Torelli-type theorems for gravitational instantons with quadratic volume growth}, 
Duke Math. J. 173 (2) 227 - 275  (2024) 
\href{https://arxiv.org/abs/2112.07504}{[arXiv:2112.07504 [math.DG]]}

\bibitem{Cherkis1}
S.~A.~Cherkis and N.~J.~Hitchin,
{\em Gravitational instantons of type $D_k$},
Commun. Math. Phys. \textbf{260} (2005), 299-317
\href{https://arxiv.org/abs/hep-th/0310084}{[arXiv:hep-th/0310084 [hep-th]].}

\bibitem{Cherkis2}
S.~A.~Cherkis and A.~Kapustin,
{\em $D_k$ gravitational instantons and Nahm equations},
Adv. Theor. Math. Phys. \textbf{2} (1999), 1287-1306
\href{https://arxiv.org/abs/hep-th/9803112}{[arXiv:hep-th/9803112 [hep-th]].}

\bibitem{Giribet}
C.~Corral, G.~Giribet and R.~Olea,
{\em Self-dual gravitational instantons in conformal gravity: Conserved 
charges and thermodynamics},
Phys. Rev. D \textbf{104} (2021) no.6, 064026
\href{https://arxiv.org/abs/2105.10574}{[arXiv:2105.10574 [hep-th]]}.

\bibitem{Derdzinski}
A.~Derdzi\'nski, 
{\em Self-dual K\"ahler manifolds and Einstein manifolds of dimension four}, 
Compositio Math., 49 (1983), pp. 405-433.

\bibitem{DunajskiBook} 
  M.~Dunajski,
  {\em Solitons, instantons, and twistors}
  (Oxford graduate texts in mathematics. 19)

\bibitem{Dunajskietal1}
M.~Dunajski, J.~Gutowski, W.~Sabra and P.~Tod,
{\em Cosmological Einstein-Maxwell Instantons and Euclidean Supersymmetry: Anti-Self-Dual Solutions},
Class. Quant. Grav. \textbf{28} (2011), 025007
\href{https://arxiv.org/abs/1006.5149}{[arXiv:1006.5149 [hep-th]]}.

\bibitem{Dunajskietal2}
M.~Dunajski, J.~B.~Gutowski, W.~A.~Sabra and P.~Tod,
{\em Cosmological Einstein-Maxwell Instantons and Euclidean Supersymmetry: Beyond Self-Duality},
JHEP \textbf{03} (2011), 131
\href{https://arxiv.org/abs/1012.1326}{[arXiv:1012.1326 [hep-th]]}.

\bibitem{Dunajski06}
M.~Dunajski and S.~A.~Hartnoll,
{\em Einstein-Maxwell gravitational instantons and five dimensional solitonic strings},
Class. Quant. Grav. \textbf{24} (2007), 1841-1862
\href{https://arxiv.org/abs/hep-th/0610261}{[arXiv:hep-th/0610261 [hep-th]]}.

\bibitem{DunajskiTod}
M.~Dunajski and P.~Tod,
{\em Four--Dimensional Metrics Conformal to Kahler},
Math. Proc. Cambridge Phil. Soc. \textbf{148} (2010), 485
\href{https://arxiv.org/abs/0901.2261}{[arXiv:0901.2261 [math.DG]]}.

\bibitem{DunajskiTod2014}
M.~Dunajski and P.~Tod,
{\em Self-Dual Conformal Gravity},
Commun. Math. Phys. \textbf{331} (2014), 351-373
\href{https://arxiv.org/abs/1304.7772}{[arXiv:1304.7772 [hep-th]]}.

\bibitem{Flaherty78}
E.~J.~Flaherty, 
{\em The nonlinear graviton in interaction with a photon}, 
Gen. Relat. Gravit. {\bf 9}, 961-978 (1978)

\bibitem{Frolov}
V.~P.~Frolov, P.~Krtous and D.~Kubiznak,
{\em Black holes, hidden symmetries, and complete integrability},
Living Rev. Rel. \textbf{20} (2017) no.1, 6
\href{https://arxiv.org/abs/1705.05482}{[arXiv:1705.05482 [gr-qc]].}

\bibitem{Gibbons1980}
G.~W.~Gibbons, 
{\em Gravitational instantons: A survey}, in: Osterwalder, K. (eds) Mathematical Problems in Theoretical Physics. Lecture Notes in Physics, vol 116. Springer, Berlin, Heidelberg (1980)

\bibitem{GH1979}
G.~W.~Gibbons and S.~W.~Hawking,
{\em Classification of Gravitational Instanton Symmetries},
Commun. Math. Phys. \textbf{66} (1979), 291-310

\bibitem{GPR1979}
G.~W.~Gibbons, C.~N.~Pope and H.~Romer,
{\em Index Theorem Boundary Terms for Gravitational Instantons},
Nucl. Phys. B \textbf{157} (1979), 377-386

\bibitem{Gutowski}
J.~B.~Gutowski and W.~A.~Sabra,
{\em Gravitational Instantons and Euclidean Supersymmetry},
Phys. Lett. B \textbf{693} (2010), 498-502
\href{https://arxiv.org/abs/1007.2421}{[arXiv:1007.2421 [hep-th]]}.

\bibitem{Hawking1976}
S.~W.~Hawking,
{\em Gravitational Instantons},
Phys. Lett. A \textbf{60} (1977), 81

\bibitem{Hawking1978}
S.~W.~Hawking,
{\em Space-Time Foam},
Nucl. Phys. B \textbf{144} (1978), 349-362

\bibitem{HPP1979}
S.~W.~Hawking, D.~N.~Page and C.~N.~Pope,
{\em Quantum Gravitational Bubbles},
Nucl. Phys. B \textbf{170} (1980), 283-306

\bibitem{Hein1}
H.-J.~Hein, 
{\em Gravitational instantons from rational elliptic surfaces}, 
J. Amer. Math. Soc. 25 (2012), no. 2, 355--393

\bibitem{HeinLebrun}
H.-J.~Hein and C.~LeBrun, 
{\em Mass in K\"ahler geometry}, 
Comm. Math. Phys., 347 (2016), pp. 183--221 
\href{https://arxiv.org/abs/1507.08885}{[arXiv:1507.08885 [math.DG]]}

\bibitem{Hein2}
H.-J.~Hein, S.~Sun, J.~Viaclovsky, and R.~Zhang, 
{\em Nilpotent structures and collapsing Ricci-flat metrics on the K3 surface}, 
J. Amer. Math. Soc., 35 (2022), pp. 123--209 
\href{https://arxiv.org/abs/1807.09367}{[arXiv:1807.09367 [math.DG]]}

\bibitem{Koca}
C.~Koca, C.~W.~T\o nnesen-Friedman, 
{\em Strongly Hermitian Einstein-Maxwell solutions on ruled surfaces}, 
Ann. Glob. Anal. Geom. {\bf 50}, 29-46 (2016)
\href{https://arxiv.org/abs/1511.06805}{[arXiv:1511.06805 [math.DG]]}

\bibitem{Kronheimer}
P.~B.~Kronheimer, 
{\em The construction of ALE spaces as hyper-K\"ahler quotients}, 
J. Diff. Geom. 29, no.3, 665-683 (1989)

\bibitem{Kronheimer2}
P.~B.~Kronheimer, 
{\em A Torelli-type theorem for gravitational instantons}, 
J. Diff. Geom., 29, pp. 685--697 (1989)

\bibitem{Lebrun91}
C.~LeBrun, 
{\em Explicit self-dual metrics on $\mathbb{CP}^{2}\#...\#\mathbb{CP}^{2}$},
J. Differential Geom. 34(1), 223-253, (1991)

\bibitem{Lebrun2010}
C.~LeBrun, 
{\em The Einstein-Maxwell equations, extremal K\"ahler metrics, and Seiberg-Witten theory}, 
in The Many Facets of Geometry, Oxford Univ. Press, Oxford, 2010, pp. 17--33 
\href{https://arxiv.org/abs/0803.3734}{[arXiv:0803.3734 [math.DG]]}

\bibitem{Lebrun2012}
C.~LeBrun,  
{\em On Einstein, Hermitian 4-manifolds}, 
Journal of Differential Geometry, 90(2), 277-302 (2012) 
\href{https://arxiv.org/abs/1010.0238}{[arXiv:1010.0238 [math.DG]]}

\bibitem{Lebrun2015}
C.~LeBrun, 
{\em The Einstein-Maxwell equations, K\"ahler metrics, and Hermitian geometry}, J. Geom. Phys., 91 (2015), pp. 163--171
\href{https://arxiv.org/abs/1411.3992}{[arXiv:1411.3992 [math.DG]]}

\bibitem{Lebrun2016}
C.~LeBrun, 
{\em The Einstein-Maxwell Equations and Conformally K\"ahler Geometry}, 
Commun. Math. Phys. 344, 621--653 (2016) 
\href{https://arxiv.org/abs/1504.06669}{[arXiv:1504.06669 [math.DG]]}

\bibitem{Lebrun2020}
C.~LeBrun, 
{\em Bach-Flat K\"ahler Surfaces}, 
J. Geom. Analysis 30 (2020) 2491--2514 
\href{https://arxiv.org/abs/1702.03840}{[arXiv:1702.03840 [math.DG]]}

\bibitem{Li2012}
J.~Li, H.~S.~Liu, H.~Lu and Z.~L.~Wang,
{\em Fermi Surfaces and Analytic Green's Functions from Conformal Gravity},
JHEP \textbf{02} (2013), 109 
\href{https://arxiv.org/abs/1210.5000}{[arXiv:1210.5000 [hep-th]]}.

\bibitem{Liu2012}
H.~S.~Liu and H.~Lu,
{\em Charged Rotating AdS Black Hole and Its Thermodynamics in Conformal Gravity},
JHEP \textbf{02} (2013), 139
\href{https://arxiv.org/abs/1212.6264}{[arXiv:1212.6264 [hep-th]]}.

\bibitem{Lu2012}
H.~Lu, Y.~Pang, C.~N.~Pope and J.~F.~Vazquez-Poritz,
{\em AdS and Lifshitz Black Holes in Conformal and Einstein-Weyl Gravities},
Phys. Rev. D \textbf{86} (2012), 044011
\href{https://arxiv.org/abs/1204.1062}{[arXiv:1204.1062 [hep-th]]}.

\bibitem{Lucietti}
J.~Lucietti, P.~Ntokos and S.~G.~Ovchinnikov,
{\em On the uniqueness of supersymmetric AdS(5) black holes with toric symmetry},
Class. Quant. Grav. \textbf{39} (2022) no.24, 245006
\href{https://arxiv.org/abs/2208.00896}{[arXiv:2208.00896 [hep-th]]}.

\bibitem{Minerbe}
V.~Minerbe, 
{\em Rigidity for multi-Taub-NUT metrics}, 
J. Reine Angew. Math., 656 (2011), pp. 47--58 
\href{https://arxiv.org/abs/0910.5792}{[arXiv:0910.5792 [math.DG]]}

\bibitem{PP}
D.~N.~Page and C.~N.~Pope,
{\em Inhomogeneous Einstein metrics on complex line bundles},
Class. Quant. Grav. \textbf{4} (1987), 213-225

\bibitem{PR1}
R.~Penrose and W.~Rindler,
\newblock {\em Spinors and space-time: Volume 1, Two-spinor calculus and
  relativistic fields}, volume~1.
\newblock Cambridge University Press, 1984.

\bibitem{PR2}
R.~Penrose and W.~Rindler,
\newblock {\em Spinors and space-time: Volume 2, 
Spinor and twistor methods in space-time geometry}, 
\newblock Cambridge University Press, 1986.

\bibitem{PD}
 J.~F.~Pleba\'nski and M.~Demia\'nski,  
 {\em Rotating, charged, and uniformly accelerating mass in general 
 relativity}, Annals of Physics, 98(1), 98-127 (1976)

\bibitem{Smilga}
A.~V.~Smilga,
{\em Vacuum structure in quantum gravity}, 
Nucl. Phys. B \textbf{234} (1984), 402-412

\bibitem{Strominger}
A.~Strominger, G.~T.~Horowitz and M.~J.~Perry,
{\em Instantons in Conformal Gravity},
Nucl. Phys. B \textbf{238} (1984), 653-664

\bibitem{Sun}
S.~Sun and R.~Zhang,
{\em Collapsing geometry of hyperk\"ahler 4-manifolds and applications},
\href{https://arxiv.org/abs/2108.12991}{[arXiv:2108.12991 [math.DG]].}

\bibitem{Tod95}
K.~P.~Tod, 
{\em Scalar-flat K\"ahler and hyper-K\"ahler metrics from Painlev\'e-III}, Classical and Quantum Gravity 12.6 (1995): 1535.

\bibitem{Tod95b}
K.~P.~Tod, 
{\em The $SU(\infty)$-Toda field equation and special four-dimensional metrics}, 
in Geometry and physics (Aarhus, 1995), vol. 184 of Lecture Notes in Pure and Appl. Math., Dekker, New York, 1997, pp. 307--312.

\bibitem{Tod2020}
P.~Tod,
{\em One-sided type-D Ricci-flat metrics},
\href{https://arxiv.org/abs/2003.03234}{[arXiv:2003.03234 [gr-qc]]}.

\bibitem{Tod2024}
P.~Tod,
{\em One-sided type-D metrics with aligned Einstein-Maxwell},
\href{https://arxiv.org/abs/2410.13410}{[arXiv:2410.13410 [gr-qc]]}.

\bibitem{Wald}
R.~M.~Wald, 
{\em General Relativity}, 
Chicago, Usa: Univ. Pr. ( 1984) 491p

\bibitem{WalkerPenrose}
M.~Walker and R.~Penrose,
{\em On quadratic first integrals of the geodesic equations for type [22] spacetimes},
Commun. Math. Phys. \textbf{18} (1970), 265-274

\bibitem{Wright}
E.~P.~Wright, {\em Quotients of gravitational instantons}, Ann. Global Anal. Geom., 41 (2012), pp. 91--108, \href{https://arxiv.org/abs/1102.2442}{[arXiv:1102.2442 [math.DG]].}

\end{thebibliography}
\end{document}